\documentclass{amsart}
\title{Using ciliate operations to construct chromosome phylogenies}

\author{{Jacob L. Herlin}, {Anna Nelson} and {Marion Scheepers}}
\thanks{This research was started during the Summer of 2011 when Herlin and Nelson participated in the Boise State University Mathematics R.E.U. program, funded by NSF grant DMS 1062857. We gratefully acknowledge funding by the NSF and by Boise State University. 
}
\date{January 6, 2014} 

\newtheorem{theorem}{\bf Theorem}
\newtheorem{lemma}[theorem]{\bf Lemma}

\newtheorem{definition}{\bf Definition}

\newcommand{\integers}{{\mathbb Z}}

\usepackage{framed}
\PassOptionsToPackage{svgnames}{xcolor}
\usepackage{listings}
\usepackage[usenames,dvipsnames,svgnames,table]{xcolor}
\usepackage{tikz}
\usepackage{graphicx}
\usetikzlibrary{arrows,shapes}
\usepackage{scalebar}
\usepackage[all, knot]{xy}
\usepackage{breqn}
\usepackage{arydshln}
\usepackage{subfig}
\usepackage{wrapfig}
\usepackage{amssymb}                 
\usepackage{epsfig}
\usepackage{fancybox}           
\usepackage{version}
\usepackage[rflt]{floatflt}
\usepackage{fancyhdr}
\usepackage[table]{xcolor}
\usepackage{calc}
\usepackage{ifthen}
\usetikzlibrary{positioning,shadows,arrows}
\usepackage{rotating}
\begin{document}
\begin{abstract}
Whole genome sequencing has revealed several examples where genomes of different species are related by permutation. The number of certain types of rearrangements needed to transform one permuted list into another can measure the distance between such lists. Using an algorithm based on three basic DNA editing operations suggested by a model for ciliate micro nuclear decryption, this study defines the distance between two permutations to be the number of ciliate operations the algorithm performs during such a transformation. Combining well-known clustering methods with this distance function enables one to construct corresponding phylogenies. These ideas are illustrated by exploring the phylogenetic relationships among the chromosomes of eight fruitfly (drosophila) species, using the well-known UPGMA algorithm on the distance function provided by the ciliate operations. 

\end{abstract}
\maketitle

\vspace{0.1in}

Over evolutionary 
time ``local" DNA editing events such as nucleotide substitutions, deletions or insertions diversify the set of DNA sequences present in organisms. Results of whole genome sequencing suggest that also ``global" DNA editing events 
diversify these DNA sequences. 

Consider two species $S_1$ and $S_2$ with a common ancestor whose genome was organized over $n$ linear chromosomes. A gene $G$ of the ancestor was inherited as gene $G_1$ by species $S_1$ and as gene $G_2$ by species $S_2$. $G_1$ and $G_2$ are \emph{orthologous} genes, or simply \emph{orthologs}. Assume that the species $S_1$ and $S_2$ each also has $n$ chromosomes, and that for each ancestral chromosome $i$, the orthologs of any ancestral gene on chromosome $i$ are also in the descendant species $S_1$ and $S_2$ on the corresponding chromosome $i$. This assumption is known, in the context of certain fruitfly species, as the \emph{Muller hypothesis}\footnote{Named after H.J. Muller who observed in the 1940's \cite{HJM} that for the data then known for relatives of \emph{Drosophila melanogaster} this assumption is true even fo chromosome arms.}. In this paper we shall assume the Muller hypothesis for our applications. 

It may happen that the order in which orthologs on chromosome $i$ appear in species $S_1$ is different from the order in which they appear in species $S_2$. In this case chromosome $i$ in each of these two species can be partitioned into a number, say k, of \emph{synteny blocks}\footnote{This definition of a synteny block is more restrictive than the one used in \cite{BSRXSG}: The latter allows for differences in gene order up to a certain threshold, and does not allow for single gene blocks. See the section ``An application to genome phylogenetics" for more information.}: A synteny block is a maximal list of adjacent orthologous genes 
that have the same adjacencies in the two species. In this definition of a synteny block, we permit blocks consisting of single genes. An endpoint of a synteny block is also called a \emph{breakpoint}. Synteny blocks may have opposite orientation in two species. Thus the synteny blocks of chromosome $i$ of species $S_1$ is a \emph{signed} permutation of the corresponding synteny blocks of chromosome $i$ of species $S_2$.  This phenomenon is observed in several branches in the tree of life. Figure \ref{fig:HumanMouseX} illustrates the phenomenon for 11 synteny blocks of orthologous genes in the X chromosome of human and mouse.

\begin{center}
          \begin{figure}[h]
			\begin{tikzpicture}[scale = 0.7] 		
          		\tikzstyle{every node}=[]
			\node at (3.25, 3.95) {\tiny{\bf Human X}};

			\draw [fill = {rgb: white,1; red, 1; orange, 0; yellow, 0; green, 1; blue, 0; violet, 1}] (10,3) rectangle (11.5,3.5);                   
			\draw [fill = {rgb: white,1; red, 1; orange, 0; yellow, 0; green, 1; blue, 0; violet, 0}] (8.5,3) rectangle (10,3.5);                     
			\draw [fill = {rgb: white,1; red, 0; orange, 3; yellow, 0; green, 0; blue, 0; violet, 0}] (7,3) rectangle (8.5,3.5);                        
			\draw [fill = {rgb: white,1; red, 3; orange, 0; yellow, 0; green, 0; blue, 0; violet, 0}] (5.5,3) rectangle (7,3.5);                        
			\draw [fill =  {rgb: white,1; red, 0; orange, 0; yellow, 0; green, 0; blue, 0; violet, 1}] (4,3) rectangle (5.5,3.5);                        
			\draw [fill =  {rgb: white,1; red, 1; orange, 0; yellow, 0; green, 0; blue, 0; violet, 1}] (2.5,3) rectangle (4,3.5);                        
			\draw [fill =  {rgb: white,1; red, 0; orange, 0; yellow, 0; green, 0; blue, 1; violet, 0}] (1,3) rectangle (2.5,3.5);                        
			\draw [fill =  {rgb: white,1; red, 0; orange, 1; yellow, 0; green, 3; blue, 0; violet, 0}] (-0.5,3) rectangle (1,3.5);                        
			\draw [fill =  {rgb: white,1; red, 0; orange, 0; yellow, 1; green, 0; blue, 1; violet, 0}] (-2,3) rectangle (-0.5,3.5);                       
			\draw [fill = {rgb: white,1; red, 0; orange, 1; yellow, 0; green, 0; blue, 0; violet, 0}](-3.5,3) rectangle (-2,3.5);                         
			\draw [fill = {rgb: white,1; red, 1; orange, 0; yellow, 0; green, 0; blue, 0; violet, 0}] (-5,3) rectangle (-3.5,3.5);   

			\node at (10.75, 3.25) {\color{white}\tiny${\mathbf 11}$};
			\node at (9.25, 3.25) {\color{white}\tiny${\mathbf 10}$};
			\node at (7.75, 3.25) {\color{white}\tiny${\mathbf 9}$};
			\node at (6.25, 3.25) {\color{white}\tiny${\mathbf 8}$};
			\node at (4.75, 3.25) {\color{white}\tiny${\mathbf 7}$};
			\node at (3.25, 3.25) {\color{white}\tiny${\mathbf 6}$};
			\node at (1.75, 3.25) {\color{white}\tiny${\mathbf 5}$}; 
			\node at (0.25, 3.25) {\color{white}\tiny${\mathbf 4}$};
			\node at (-1.25, 3.25) {\color{white}\tiny${\mathbf 3}$};
			\node at (-2.75, 3.25) {\color{white}\tiny${\mathbf{2}}$};
			\node at (-4.25, 3.25) {\color{white}\tiny${\mathbf 1}$};
\node(p) at (11.0,2.95){} 
edge [->, left] node[above=3pt,name=e2]{} (6.25,0.75);

\node(p) at (9.70,2.85){} 
edge [->, left] node[above=3pt,name=e2]{} (0.25,0.75);

\node(p) at (8,2.90){} 
edge [->, right] node[above=3pt,name=e2]{} (1.75,0.75);

\node(p) at (6.55,2.95){} 
edge [->, right] node[above=3pt,name=e2]{} (3.25,0.75);

\node(p) at (5,2.90){} 
edge [->, right] node[above=3pt,name=e2]{} (-2.75,0.75);

\node(p) at (3.5,2.92){} 
edge [->, left] node[above=3pt,name=e2]{} (-1.25,0.75);

\node(p) at (1.55,2.85){} 
edge [->, right] node[above=3pt,name=e2]{} (9.25,0.75);

\node(p) at (0.00,2.85){} 
edge [->, right] node[above=3pt,name=e2]{} (10.5,0.75);

\node(p) at (-1.55,2.85){} 
edge [->, right] node[above=3pt,name=e2]{} (7.75,0.75);

\node(p) at (-3,2.85){} 
edge [->, right] node[above=3pt,name=e2]{} (4.5,0.75);

\node(p) at (-4.25,3.05){} 
edge [->, right] node[above=3pt,name=e1]{} (-4.25,0.75); 

			\node at (3.25, -0.5) {\tiny{\bf Mouse X}};

			\draw [fill = {rgb: white,1; red, 0; orange, 1; yellow, 0; green, 3; blue, 0; violet, 0}] (10,0) rectangle (11.5,0.5);                   
			\draw [fill = {rgb: white,1; red, 0; orange, 0; yellow, 0; green, 0; blue, 1; violet, 0}] (8.5,0) rectangle (10,0.5);                     
			\draw [fill = {rgb: white,1; red, 0; orange, 0; yellow, 1; green, 0; blue, 1; violet, 0}] (7,0) rectangle (8.5,0.5);                        
			\draw [fill = {rgb: white,1; red, 1; orange, 0; yellow, 0; green, 1; blue, 0; violet, 1}] (5.5,0) rectangle (7,0.5);                        
			\draw [fill =  {rgb: white,1; red, 0; orange, 1; yellow, 0; green, 0; blue, 0; violet, 0}] (4,0) rectangle (5.5,0.5);                        
			\draw [fill =   {rgb: white,1; red, 3; orange, 0; yellow, 0; green, 0; blue, 0; violet, 0}] (2.5,0) rectangle (4,0.5);                        
			\draw [fill =  {rgb: white,1; red, 0; orange, 3; yellow, 0; green, 0; blue, 0; violet, 0}] (1,0) rectangle (2.5,0.5);                        
			\draw [fill =  {rgb: white,1; red, 1; orange, 0; yellow, 0; green, 1; blue, 0; violet, 0}] (-0.5,0) rectangle (1,0.5);                        
			\draw [fill = {rgb: white,1; red, 1; orange, 0; yellow, 0; green, 0; blue, 1; violet, 0} ] (-2,0) rectangle (-0.5,0.5);                       
			\draw [fill = {rgb: white,1; red, 0; orange, 0; yellow, 0; green, 0; blue, 0; violet, 1}](-3.5,0) rectangle (-2,0.5);                         
			\draw [fill = {rgb: white,1; red, 1; orange, 0; yellow, 0; green, 0; blue, 0; violet, 0}] (-5,0) rectangle (-3.5,0.5);   

			\node at (10.75, 0.25) {\color{white}\tiny${\mathbf 4}$};
			\node at (9.25, 0.25) {\color{white}\tiny${\mathbf 5}$};
			\node at (7.75, 0.25) {\color{white}\tiny${\mathbf -3}$};
			\node at (6.25, 0.25) {\color{white}\tiny${\mathbf -11}$};
			\node at (4.75, 0.25) {\color{white}\tiny${\mathbf 2}$};
			\node at (3.25, 0.25) {\color{white}\tiny${\mathbf -8}$};
			\node at (1.75, 0.25) {\color{white}\tiny${\mathbf 9}$}; 
			\node at (0.25, 0.25) {\color{white}\tiny${\mathbf -10}$};
			\node at (-1.25, 0.25) {\color{white}\tiny${\mathbf 6}$};
			\node at (-2.75, 0.25) {\color{white}\tiny${\mathbf{-7}}$};
			\node at (-4.25, 0.25) {\color{white}\tiny${\mathbf 1}$};
			\end{tikzpicture}
 \caption{The permutation between 11 synteny blocks of the human and the mouse X chromosomes. A negative symbol denotes an orientation
    change by a 180$^0$ rotation  of a synteny block. The lengths of synteny blocks are not to scale. See Figure 2 of \cite{PT}.}
  \label{fig:HumanMouseX}
\end{figure}
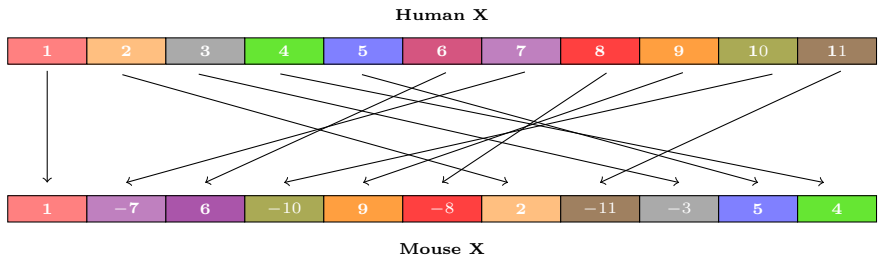
\end{center}


Since Dobshansky and Sturtevant's works \cite{DS} and \cite{SD} in the 1930's on fruitfly genomes it has been popular to use \emph{reversals}\footnote{A reversal is a rotation of a DNA segment through 180$^0$. Reversals are also called inversions.} as the primary ``global" DNA sequence editing operation to describe phylogenetic relationships among genomes. See for example \cite{BP} and \cite{HP}.

An insightful phylogenetic analysis that includes fine structural elements of reversals is given in \cite{BSRXSG}. \cite{BSRXSG} addresses the question whether a reversal can occur at arbitrary locations in the genome of an organism. Certain locations, which would disrupt the coding region of an essential gene, would not be observed in extant organisms. Similarly, locations that negatively affect the fitness of organisms would disappear over time due to ``purifying selection". Additionally, certain sequence motifs may actually promote DNA recombination that results in a genome rearrangement. For example \cite{CW} reports a correlation between \emph{breakpoints}\footnote{Referring to the Mouse X chromosome in Figure \ref{fig:HumanMouseX}, a breakpoint is a transition point between synteny blocks that are not consecutively numbered.} associated with rearrangements, and repetitive DNA. In the review \cite{DH} a similar correlation between rearrangements in bacterial genomes and repetitive DNA is discussed. These considerations suggest that genome rearrangement events that lead to the diverse genomes we observe in nature are not arbitrary, but constrained by contexts. In this paper we explore the use of \emph{context directed} DNA recombination events to analyze genome rearrangements and to construct a phylogeny based on these. 

In recent years also transpositions and block interchanges have been considered as possible ``global" DNA sequence editing operations - \cite{BP2}, \cite{CW}, \cite{MM}, \cite{YAF}. In a \emph{block interchange}, two disjoint segments of a chromosome exchange locations without changing orientation. Thus, in Figure \ref{fig:HumanMouseX}, synteny blocks 2 and 7 would have been a block interchange if synteny block 7 did not also undergo a reversal. A \emph{transposition} is a special block interchange where the two segments that exchange location are adjacent. In Figure \ref{fig:HumanMouseX} synteny blocks 4 and 5 illustrate a transposition.

On p. 1661 of \cite{BSRXSG}, in the discussion of selection of genes to which their analysis of rearrangements in fruitfly genomes apply, the authors indicate that genes deemed to have been relocated by a transposition rather than a reversal have been explicitly removed from the analysis. Thus, the analysis of \cite{BSRXSG} features reversals exclusively. On the other hand, the analysis in  \cite{CW} of rearrangements in the genomes of two nematode species includes reversals, transpositions and \emph{translocations}. A translocation occurs when segments from two different chromosomes exchange positions. In this paper we explore only reversals and block interchanges (both constrained by contexts) in the analysis of rearrangements.

Experimental results from ciliate laboratories present us with examples of DNA editing operations that routinely occur during developmental processes in these organisms. The textbook \cite{EHPPR} and the two surveys \cite{Prescott} and \cite{Prescott2} give a good starting point for information about these ``ciliate operations" and the corresponding biological background. We shall call the yet to be fully identified system in ciliates that accomplishes micro nuclear decryption\footnote{Some details regarding this process are given below in Section 1.}, the \emph{ciliate decryptome}. 

We shall illustrate how to use ``ciliate operations" to deduce potential phylogenetic relationships from genome rearrangement phenomena. Previous work, including \cite{BP}, \cite{BSRXSG} and \cite{HP}, used unconstrained reversals to deduce phylogenetic relationships. Our main ideas are to use ciliate genomic elements to model two genomes related by permutations of locations and orientations of synteny blocks,  
to apply the context directed DNA operations of the ciliate decryptome to define a distance function between the relevant permuted genomes, and to then use a classical distance based algorithm to derive phylogenies. Of the several different distance based algorithms available we selected the UPGMA algorithm\footnote{Descriptions of UPGMA can be found in the online Chapter 27 of \cite{BBEGP}, or in the textbook \cite{CB}.}. 

Then we apply these ideas to chromosomes of eight species of fruit flies (\emph{Drosophila}) to obtain a phylogeny for each of these chromosomes. 

The use of ciliate operations as the basis for deriving a distance function has the 
attractive feature that the ciliate decryptome is programmable \cite{NDL}, and the computational steps taken by the decryptome can be monitored under laboratory conditions \cite{MZC}. Thus, there are extant organisms that are poised to be employed as DNA computing devices naturally equipped to determine phylogenetic relationships among permuted genomes. 

Our paper is organized as follows: In Section 1 we briefly describe ciliate nuclear duality. This duality is the basis for modeling pairs of genomes related by permutation as genetic elements of the ciliate genome. In Section 2 we briefly describe the context directed DNA operations of the ciliate decryptome. 
In Section 3 we introduce and analyze  
the mathematical notion of a pointer list.  
In Section 4 we model relevant features of the ciliate decryptome's DNA operations by mathematical operations on pointer lists. 
In Section 5 we describe an algorithm which we call the HNS algorithm, that uses these operations on pointer lists to compute the distance between chomosomes that are related by permutation. 
In Section 6 we use data downloaded from flybase.org and the HNS- and UPGMA algorithms to construct phylogenies over eight species for each of the fruitfly chromosomes. In the closing Section 7 we discuss possible future directions related to this work.

\section{Ciliates and nuclear duality.}

A ciliate is a single cell eukaryote that hosts two types of nuclei: one type, the macro nucleus, contains the transcriptionally active somatic genome, while the other type, the micro nucleus, contains a transcriptionally silent germline-like genome. The micro nuclear genome is, in the technical sense of the word, an encrypted version of the macro nuclear genome. Special events in the ciliate life cycle predictably trigger conjugation between a pair of mating-compatible cells. Conjugation results in what amounts to a Diffie-Hellman exchange\footnote{A Diffie-Hellman exchange is a cryptographic protocol for secure exchange of a secret key in a hostile environment. The conjugants exchange a haploid copy of the germline genome, which is an encrypted version of the somatic genome.} between two conjugants, the formation of a new micro nucleus in each, and the decryption of one of more copies of the new micro nuclear genome to establish a replacement macro nuclear genome, while in each conjugant the instances of its pre-existing genome are discarded. Readers interested in a thorough survey of ciliate nuclear duality could consult \cite{Prescott}. 

\begin{center}{\bf The relationship between micro and macro nuclear DNA}\end{center}

To describe the experimentally observed relationship between the micro nuclear and macro nuclear DNA molecules, consider Figure \ref{fig:MICDiagram3}:

 \begin{center}
          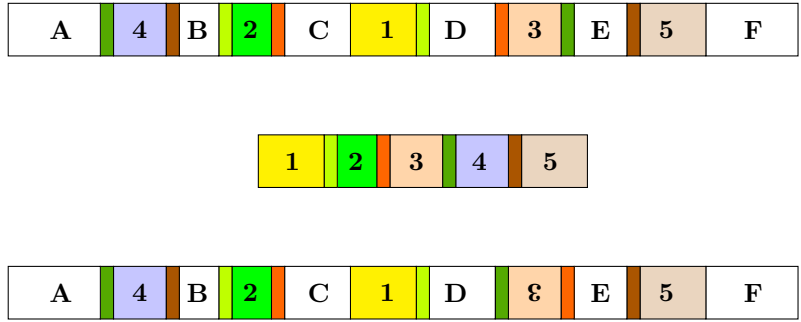
\begin{figure}[h]
			\begin{tikzpicture}[scale = 0.7] 		
          		\tikzstyle{every node}=[]
			\draw [fill = white] (8.25,5) rectangle (10.0,6);                     
			\draw [fill = {rgb: brown,1; blue,0;white,2}] (7.0,5) rectangle (8.25,6);
			\draw [fill =  {rgb: red,2; blue,0;green,1}] (6.75,5) rectangle (7.0,6);
			\draw [fill = white] (5.75,5) rectangle (6.75,6);
			\draw [fill =  {rgb: red,1; blue, 0; green, 2}] (5.5,5) rectangle (5.75,6);
			\draw [fill = {rgb: orange,1; white,2}] (4.5,5) rectangle (5.5,6);
			\draw [fill = {rgb: red,3; blue,0;yellow,2}] (4.25,5) rectangle (4.5,6);
			\draw [fill = white] (3.00,5) rectangle (4.25,6);
			\draw [fill = {rgb: yellow,3; blue,0;green,1}] (2.75,5) rectangle (3.00,6);			
			\draw [fill = yellow] (1.5,5) rectangle (2.75,6);
			\draw [fill = white] (0.25,5) rectangle (1.5,6);
			\draw [fill = {rgb: red,3; blue,0;yellow,2}] (0,5) rectangle (.25,6);			
			\draw [fill = green] (-0.75,5) rectangle (0,6);
			\draw [fill = {rgb: yellow,3; blue,0;green,1}] (-1,5) rectangle (-0.75,6);			
			\draw [fill = white] (-1.75,5) rectangle (-1,6);
			\draw [fill = {rgb: red,2; blue,0;green,1}] (-2,5) rectangle (-1.75,6);
			\draw [fill = {rgb: blue,2;white,7}] (-3,5) rectangle (-2,6);
			\draw [fill = {rgb: red,1; blue, 0; green, 2}](-3.25,5) rectangle (-3,6); 
			\draw [fill = white] (-5,5) rectangle (-3.25,6);

			\node at (9.15, 5.5) {${\mathbf F}$};
			\node at (7.5, 5.5) {${\mathbf 5}$};
			\node at (6.25, 5.5) {${\mathbf E}$};
			\node at (5.0, 5.5) {${\mathbf 3}$};
			\node at (3.5, 5.5) {${\mathbf D}$};
			\node at (2.2, 5.5) {${\mathbf 1}$};
			\node at (0.9, 5.5) {${\mathbf C}$}; 
			\node at (-0.4, 5.5) {${\mathbf 2}$};
			\node at (-1.4, 5.5) {${\mathbf B}$};
			\node at (-2.5, 5.5) {${\mathbf{4}}$};
			\node at (-4.0, 5.5) {${\mathbf A}$};
			\draw [fill = {rgb: brown,1; blue,0;white,2}] (4.75,2.5) rectangle (6.00,3.5);			
			\draw [fill = {rgb: red,2; blue,0;green,1}] (4.5,2.5) rectangle (4.75,3.5); 
			\draw [fill = {rgb: blue,2;white,7}] (3.5,2.5) rectangle (4.5,3.5); 
			\draw [fill = {rgb: red,1; blue, 0; green, 2}] (3.25,2.5) rectangle (3.5,3.5);
			\draw [fill = {rgb: orange,1; white,2}] (2.25,2.5) rectangle (3.25,3.5);
			\draw [fill = {rgb: red,3; blue,0;yellow,2}] (2,2.5) rectangle (2.25,3.5);			
			\draw [fill = green] (1.25,2.5) rectangle (2,3.5);
			\draw [fill = {rgb: yellow,3; blue,0;green,1}] (1,2.5) rectangle (1.25,3.5);
			\draw [fill = yellow] (-0.25,2.5) rectangle (1,3.5);

			\node at (5.3, 3.0) {${\mathbf 5}$};
			\node at (3.95, 3.0) {${\mathbf 4}$};
			\node at (2.75, 3.0) {${\mathbf 3}$};
			\node at (1.6, 3.0) {${\mathbf 2}$};
			\node at (0.4, 3.0) {${\mathbf{1}}$};
			\draw [fill = white] (8.25,0) rectangle (10.0,1);                     
			\draw [fill = {rgb: brown,1; blue,0;white,2}] (7.0,0) rectangle (8.25,1);
			\draw [fill =  {rgb: red,2; blue,0;green,1}] (6.75,0) rectangle (7.0,1);
			\draw [fill = white] (5.75,0) rectangle (6.75,1);
			\draw [fill = {rgb: red,3; blue,0;yellow,2}] (5.5,0) rectangle (5.75,1);
			\draw [fill = {rgb: orange,1; white,2}] (4.5,0) rectangle (5.5,1);
			\draw [fill =  {rgb: red,1; blue, 0; green, 2}] (4.25,0) rectangle (4.5,1);
			\draw [fill = white] (3.00,0) rectangle (4.25,1);
			\draw [fill = {rgb: yellow,3; blue,0;green,1}] (2.75,0) rectangle (3.00,1);			
			\draw [fill = yellow] (1.5,0) rectangle (2.75,1);
			\draw [fill = white] (0.25,0) rectangle (1.5,1);
			\draw [fill = {rgb: red,3; blue,0;yellow,2}] (0,0) rectangle (.25,1);			
			\draw [fill = green] (-0.75,0) rectangle (0,1);
			\draw [fill = {rgb: yellow,3; blue,0;green,1}] (-1,0) rectangle (-0.75,1);			
			\draw [fill = white] (-1.75,0) rectangle (-1,1);
			\draw [fill = {rgb: red,2; blue,0;green,1}] (-2,0) rectangle (-1.75,1);
			\draw [fill = {rgb: blue,2;white,7}] (-3,0) rectangle (-2,1);
			\draw [fill = {rgb: red,1; blue, 0; green, 2}](-3.25,0) rectangle (-3,1); 
			\draw [fill = white] (-5,0) rectangle (-3.25,1);

			\node at (9.15, 0.5) {${\mathbf F}$};
			\node at (7.5, 0.5) {${\mathbf 5}$};
			\node at (6.25, 0.5) {${\mathbf E}$};
			\node at (5.0, 0.5) {$\rotatebox[origin=c]{180}{\bf 3}$}; 
			\node at (3.5, 0.5) {${\mathbf D}$};
			\node at (2.2, 0.5) {${\mathbf 1}$};
			\node at (0.9, 0.5) {${\mathbf C}$}; 
			\node at (-0.4, 0.5) {${\mathbf 2}$};
			\node at (-1.4, 0.5) {${\mathbf B}$}; 
			\node at (-2.5, 0.5) {${\mathbf{4}}$};
			\node at (-4.0, 0.5) {${\mathbf A}$};

			\end{tikzpicture}
  \caption{The top diagram depicts a possible micro nuclear precursor, and the bottom diagram is another possible micro nuclear precursor of the macro nuclear gene in the middle diagram.}
  \label{fig:MICDiagram3}
\end{figure}
\end{center}

The micro nuclear DNA sequences in the top and the bottom rows of Figure \ref{fig:MICDiagram3} each has three types of regions: The white blocks, labeled with letters, are called \emph{internal eliminated sequences} (IESs). The blocks labeled with numbers are called \emph{macro nuclear destined sequences} (MDSs), while the narrow strips are called \emph{pointers}. As the micro nuclear precursors show,  
there are two copies of each pointer: For example MDS 2 has a pointer on the left flank that is identical to the pointer on the right flank of MDS 1. This pointer will be called the ``1-2 \emph{pointer}". And MDS 2 has a pointer on its right flank which is identical to the pointer on the left flank of MDS 3. This pointer is called the ``2-3 \emph{pointer}". The other pointers are named similarly. Also note that MDS 1 does not have a pointer on its left flank, and MDS 5 does not have a pointer on its right flank. As MDS3 and the pointers on its flanks show in the bottom row of Figure \ref{fig:MICDiagram3}, in the micro nuclear precursor an MDS plus its flanking pointer(s), as a unit, can be in a 180-degree rotated orientation of the corresponding components in the macro nuclear gene. 
The corresponding macro nuclear sequence in the middle row of Figure \ref{fig:MICDiagram3} contains only one of each of  
the pointers present in its micro nuclear precursor, and all the MDSs, but none of the IESs of the micro nuclear precursor. In the macro nuclear sequence these components occur in a specific order, which we call the \emph{canonical order}. 

In ``shorthand" the micro nuclear precursor in the top row of Figure \ref{fig:MICDiagram3} is
  $\lbrack 4,\, 2,\ 1,\ 3,\ 5\rbrack$
while the micro nuclear precursor in the bottom row of Figure \ref{fig:MICDiagram3} is
  $\lbrack 4,\, 2,\ 1,\ -3,\ 5\rbrack$.

\section{The ciliate DNA operations}

We now turn to the ciliate algorithm that processes micro nuclear precursors to produce their corresponding macro nuclear versions. The journal articles \cite{AJSL} and \cite{PER} propose hypotheses about biochemical processes that perform the decryption algorithm in ciliates. 
We do not examine the biochemical foundations here.

Textbook \cite{EHPPR} describes three DNA editing operations underlying this decryption process. There is  
experimental evidence 
that these three operations  
accomplish the decryption process. 
The journal article \cite{MZC} gives 
experimental data about the DNA products of intermediate steps of the ciliate algorithm.  
We henceforth 
assume 
that the three operations that produce macro nuclear molecules from their micro nuclear precursors are as proposed in \cite{EHPPR}: context directed block interchanges (swaps), context directed reversals and context directed excisions. 

{\flushleft{\bf \underline{Context directed block interchanges (swaps):}}}  The top strip in Figure \ref{fig:dladresult} represents a segment of DNA in a micro nuclear chromosome of some ciliate. The symbols ${\mathbf p}$ and ${\mathbf q}$ denote identified pointers, 
while 
$A$, $B$, $M$, $X$ and $Y$ 
represent segments of DNA. The three necessary 
conditions to swap segments $X$ and $Y$ are:
{\flushleft{\bf 1}} $X$ and $Y$ both have an occurrence of each of the pointers ${\mathbf p}$ and ${\mathbf q}$ at their flanks;  
{\flushleft{\bf 2}} The  pointer pair ${\mathbf p},\, {\mathbf q}$ appears in the (alternating) context 
$
  \cdots {\mathbf p}\cdots {\mathbf q}\cdots {\mathbf p}\cdots {\mathbf q} \cdots;
$
{\flushleft{\bf 3}} Neither occurrence of the pointer ${\mathbf p}$ or of pointer ${\mathbf q}$ is flanked by a pair of successively numbered MDSs.




 \begin{center}
          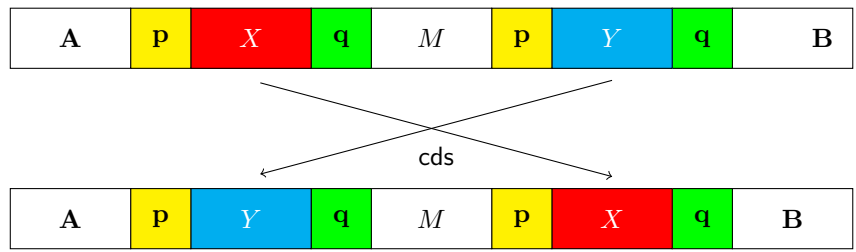
\begin{figure}[h]
			\begin{tikzpicture}[scale = 0.8] 		
          		\tikzstyle{every node}=[]
		                           
			\draw [fill = white] (7,0) rectangle (9,1);
			\draw [fill = green] (6,0) rectangle (7,1);
			\draw [fill = red] (4,0) rectangle (6,1);
			\draw [fill = yellow] (3,0) rectangle (4,1);
			\draw [fill = white] (1,0) rectangle (3,1);
			\draw [fill = green] (0,0) rectangle (1,1);
			\draw [fill = cyan] (-2,0) rectangle (0,1);
			\draw [fill = yellow] (-3,0) rectangle (-2,1);
			\draw [fill = white] (-5,0) rectangle (-3,1);

			\node at (8.0, 0.5) {${\mathbf B}$};
			\node at (6.5, 0.5) {${\mathbf q}$};
			\node at (5.0, 0.5) {${\color{white}X}$};
			\node at (3.5, 0.5) {${\mathbf p}$};
			\node at (2.0, 0.5) {$ {M}$}; 
			\node at (0.5, 0.5) {${\mathbf q}$};
			\node at (-1.0, 0.5) {${\color{white}Y}$};
			\node at (-2.5, 0.5) {${\mathbf{p}}$};
			\node at (-4.0, 0.5) {${\mathbf A}$};

\node(p) at (-1.0,1.2){}
edge [<-, right] node[above=3pt,name=e1]{} (5.0,2.8); 

\node(ps) at (-1.0,2.8){}
edge [->, left] node[below=3pt,name=e2] {{\sf cds}} (5.0,1.2);

                     
			\draw [fill = white] (7,3) rectangle (9,4);
			\draw [fill = green] (6,3) rectangle (7,4);
			\draw [fill = cyan] (4,3) rectangle (6,4);
			\draw [fill = yellow] (3,3) rectangle (4,4);
			\draw [fill = white] (1,3) rectangle (3,4);
			\draw [fill = green] (0,3) rectangle (1,4);
			\draw [fill = red] (-2,3) rectangle (0,4);
			\draw [fill = yellow] (-3,3) rectangle (-2,4);
			\draw [fill = white] (-5,3) rectangle (-3,4);

			\node at (8.5, 3.5) {${\mathbf B}$};
			\node at (6.5, 3.5) {${\mathbf q}$};
			\node at (5.0, 3.5) {${\color{white}Y}$};
			\node at (3.5, 3.5) {${\mathbf p}$};
			\node at (2.0, 3.5) {${M}$};
			\node at (0.5, 3.5) {${\mathbf q}$};
			\node at (-1.0, 3.5) {${\color{white}X}$};
			\node at (-2.5, 3.5) {${\mathbf{p}}$};
			\node at (-4.0, 3.5) {${\mathbf A}$};

			\end{tikzpicture}
\vspace{0.05in}
  \caption{Context Directed Block Swaps: The ${\mathbf p}\cdots{\mathbf q}\cdots{\mathbf p}\cdots{\mathbf q}$ pointer context permits swapping the DNA segments X and Y if M, X and Y meet requirement 3.}
  \label{fig:dladresult}
    \end{figure}
\end{center}
Only when \emph{all three} conditions are met is an interchange of the segments $X$ and $Y$ permitted. The result of this swap is depicted in the bottom strip of Figure \ref{fig:dladresult}.
The reader may check that subsequent to an application of {\sf cds} the contextual conditions 1 and 2 are still valid, but condition 3 is no longer met: Indeed, one occurrence of each of the pointers ${\mathbf p}$ and ${\mathbf q}$ is now flanked by successively numbered MDSs.
Figure \ref{fig:cdsDiagram2} gives a specific example to illustrate the last point:

 \begin{center}
          \begin{figure}[ht]
			\begin{tikzpicture}[scale = 0.7] 		
          		\tikzstyle{every node}=[]
			\draw [fill = white] (8.25,5) rectangle (9.0,6);                     
			\draw [fill = {rgb: brown,1; blue,0;white,2}] (7.5,5) rectangle (8.25,6);
			\draw [fill =  {rgb: red,2; blue,0;green,1}] (7.25,5) rectangle (7.5,6);
			\draw [fill = white] (6.5,5) rectangle (7.25,6);
			\draw [fill =  {rgb: red,1; blue, 0; green, 2}] (6.25,5) rectangle (6.5,6);
			\draw [fill = {rgb: orange,1; white,2}] (5.5,5) rectangle (6.25,6);
			\draw [fill = {rgb: red,3; blue,0;yellow,2}] (5.25,5) rectangle (5.5,6);
			\draw [fill = white] (4.50,5) rectangle (5.25,6);
			\draw [fill = {rgb: yellow,3; blue,0;green,1}] (4.25,5) rectangle (4.50,6);		
			\draw [fill = yellow] (3.5,5) rectangle (4.25,6);
			\draw [fill = white] (2.75,5) rectangle (3.5,6);
			\draw [fill = {rgb: red,2; blue,0;green,1}] (2.50,5) rectangle (2.75,6); 
			\draw [fill = {rgb: blue,2;white,7}] (1.75,5) rectangle (2.5,6); 
			\draw [fill = {rgb: red,1; blue, 0; green, 2}] (1.5,5) rectangle (1.75,6); 
			\draw [fill = white] (0.75,5) rectangle (1.5,6); 
			\draw [fill =  {rgb: red,3; blue,0;yellow,2}] (0.5,5) rectangle (0.75,6); 
			\draw [fill = green] (-0.25,5) rectangle (0.5,6); 
			\draw [fill =  {rgb: yellow,3; blue,0;green,1}](-0.5,5) rectangle (-0.25,6); 
			\draw [fill = white] (-1.25,5) rectangle (-0.5,6);
			
			\draw [thick,color=green] (-0.2,4.8) -- (1.5,4.8);
			\draw [thick,color=orange] (4.6,4.8) -- (6.2,4.8);

			\node at (8.65, 5.5) {${\mathbf F}$};
			\node at (7.8, 5.5) {${\mathbf 5}$};
			\node at (6.8, 5.5) {${\mathbf E}$};
			\node at (5.9, 5.5) {${\mathbf 3}$};
			\node at (4.9, 5.5) {${\mathbf D}$};
			\node at (3.9, 5.5) {${\mathbf 1}$};
			\node at (3.1, 5.5) {${\mathbf C}$}; 
			\node at (2.1, 5.5) {${\mathbf 4}$}; %
			\node at (1.1, 5.5) {${\mathbf B}$};
			\node at (0.1, 5.5) {${\mathbf{2}}$}; %
			\node at (-0.9, 5.5) {${\mathbf A}$};

			\draw [fill = white] (8.25,0) rectangle (9.0,1);                     
			\draw [fill = {rgb: brown,1; blue,0;white,2}] (7.5,0) rectangle (8.25,1);
			\draw [fill =  {rgb: red,2; blue,0;green,1}] (7.25,0) rectangle (7.5,1);
			\draw [fill = white] (6.5,0) rectangle (7.25,1);
			\draw [fill = {rgb: red,3; blue,0;yellow,2}] (6.25,0) rectangle (6.5,1);
			\draw [fill = white] (5.5,0) rectangle (6.25,1);
			\draw [fill =  {rgb: red,1; blue, 0; green, 2}] (5.25,0) rectangle (5.5,1);
			\draw [fill = green] (4.5,0) rectangle (5.25,1);
			\draw [fill = {rgb: yellow,3; blue,0;green,1}] (4.25,0) rectangle (4.5,1);		
			\draw [fill = yellow] (3.5,0) rectangle (4.25,1);
			\draw [fill = white] (2.75,0) rectangle (3.5,1);
			\draw [fill = {rgb: red,2; blue,0;green,1}] (2.5,0) rectangle (2.75,1);	
			\draw [fill = {rgb: blue,2;white,7}] (1.75,0) rectangle (2.5,1);
			\draw [fill = {rgb: red,1; blue, 0; green, 2}] (1.5,0) rectangle (1.75,1);		
			\draw [fill = {rgb: orange,1; white,2}] (0.75,0) rectangle (1.5,1);
			\draw [fill =  {rgb: red,3; blue,0;yellow,2}] (0.5,0) rectangle (0.75,1);			
			\draw [fill = white] (-0.25,0) rectangle (0.5,1);
			\draw [fill =  {rgb: yellow,3; blue,0;green,1}](-0.5,0) rectangle (-0.25,1); 
			\draw [fill = white] (-1.25,0) rectangle (-0.5,1);
			
			\draw [thick,color=orange] (-0.2,1.2) -- (1.5,1.2);
			\draw [thick,color=green] (4.6,1.2) -- (6.2,1.2);
			
			\node at (8.65, 0.5) {${\mathbf F}$};
			\node at (7.8, 0.5) {${\mathbf 5}$};
			\node at (6.8, 0.5) {${\mathbf E}$};
			\node at (5.9, 0.5) {${\bf B}$}; 
			\node at (4.9, 0.5) {${\mathbf 2}$};
			\node at (3.9, 0.5) {${\mathbf 1}$};
			\node at (3.1, 0.5) {${\mathbf C}$};
			\node at (2.1, 0.5) {${\mathbf 4}$};
			\node at (1.1, 0.5) {${\mathbf 3}$}; 
			\node at (0.1, 0.5) {${\mathbf{D}}$};
			\node at (-0.9, 0.5) {${\mathbf A}$};
			
			\draw [thick,->] (0.65,4.7) -- (5.4,1.3);
			\draw [thick,->] (5.4,4.7) -- (0.65,1.3);
			
			\node at (1.7,4.7){\bf q};
			\node at (-0.4,4.7){\bf p};
			\node at (6.4,4.7){\bf q};
			\node at (4.4,4.7){\bf p};

			\node at (1.7,1.2){\bf q};
			\node at (-0.4,1.2){\bf p};
			\node at (6.4,1.2){\bf q};
			\node at (4.4,1.2){\bf p};
			
			\end{tikzpicture}
  \caption{The top diagram depicts a possible micro nuclear precursor, and the bottom diagram is the result of {\sf cds} applied to the pointer pair {\bf p}=(1,2) and {\bf q}=(3,4).}
  \label{fig:cdsDiagram2}
\end{figure}
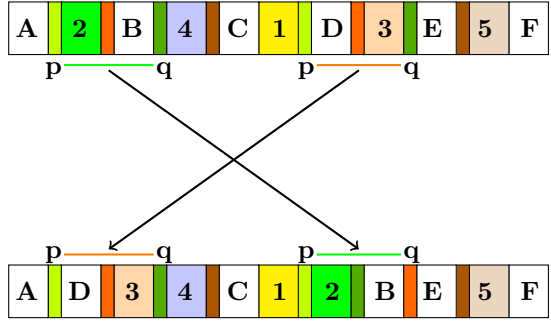
\end{center}

{\flushleft{\bf \underline{Context directed reversal:}}} To describe a context directed reversal, consider the left strip in Figure \ref{fig:hiresult}. It is a depiction of a segment of DNA appearing in the micro nucleus. To rotate the yellow segment, labeled by an upside-down $A$, by $180^{0}$, that is, to ``reverse $A$", two neccessary contextual conditions must be met: 
{\flushleft{\bf 1}} $A$ is flanked by a pointer ${\mathbf p}$ on one end, and by the 180$^0$ rotation\footnote{In text the 180$^0$ rotation of ${\mathbf p}$ will be denoted $-{\mathbf p}$.} of ${\mathbf p}$ on the other end;
{\flushleft{\bf 2}} Neither occurrence of ${\mathbf p}$ is flanked by successively numbered MDSs. 

When both of these contextual requirements are met rotation of the segment labelled $A$ through $180^{0}$ is permitted. 
The result is of this context directed reversal is depicted by the right strip in Figure \ref{fig:hiresult}. 

 \begin{center}
          \begin{figure}[h]
			\begin{tikzpicture}[scale = 0.7] 		
          		\tikzstyle{every node}=[]                           
			\draw [fill = red] (3,0) rectangle (4,1);
			\draw [fill = pink] (2,0) rectangle (3,1);
			\draw [fill = yellow] (0,0) rectangle (2,1);
			\draw [fill = pink] (-1,0) rectangle (0,1);
			\draw [fill = cyan] (-2,0) rectangle (-1,1);

			\node at (3.5, 0.5) {${\color{white}Y}$};
			\node at (2.5, 0.5) {${\mathbf p}$};
			\node at (1.0, 0.5) {$ \rotatebox[origin=c]{180}{A}$}; 
			\node at (-0.5, 0.5) {${\mathbf d}$};
			\node at (-1.5, 0.5) {${\color{white}X}$};

			\draw [line width = 1, ->]  (4.2, 0.5)  to (5.8,0.5);
			\draw node at (5, 0.9) {{\sf cdr}};


			\draw [fill = red] (11,0) rectangle (12,1);
			\draw [fill = pink] (10,0) rectangle (11,1);
			\draw [fill = yellow] (8,0) rectangle (10,1);
			\draw [fill = pink] (7,0) rectangle (8,1);
			\draw [fill = cyan] (6,0) rectangle (7,1);

			\node at (11.5, 0.5) {${\color{white}Y}$};
			\node at (10.5, 0.5) {${\mathbf p}$};
			\node at (9.0, 0.5) {${A}$};
			\node at (7.5, 0.5) {${\mathbf d}$};
			\node at (6.5, 0.5) {${\color{white}X}$};

			\end{tikzpicture}
\vspace{0.05in}
  \caption{Context Directed Reversal: The {\bf -p}...{\bf p} or {\bf p}...{\bf -p} pointer context permits rotating the segment A flanked by them through 180 degrees if condition 2 is met by X, A and Y.}
  \label{fig:hiresult}
    \end{figure}
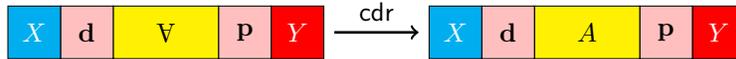
\end{center}

As the reader may check, subsequent to a context directed reversal, one of the occurrences of the pointer ${\mathbf p}$ now has successively numbered MDSs on both flanks and no further applications of {\sf cdr} are permitted to this pointer context .

{\flushleft{\bf \underline{Context directed excision:}}} To describe context directed excision consider Figure \ref{fig:cdediagram}. In it the pointer ${\mathbf p}$ flanks a DNA segment identified as an IES (the yellow segment). This context ${\mathbf p}\, {\sf IES}\,{\mathbf p}$ permits the excision of the {\sf IES} segment plus one of the pointers, provided that each occurrence of {\bf p} is flanked by an MDS. The result of {\sf cde} is the joining the DNA segments flanking the original pair of pointers, to the flanks of the remaining pointer.

\vspace{0.1in}
 \begin{center}
          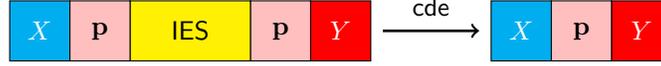
\begin{figure}[h]
			\begin{tikzpicture}[scale = 0.8] 		
          		\tikzstyle{every node}=[]  
                     \vspace{0.1in}
                           
			\draw [fill = red] (3,0) rectangle (4,1);
			\draw [fill = pink] (2,0) rectangle (3,1);
			\draw [fill = yellow] (0,0) rectangle (2,1);
			\draw [fill = pink] (-1,0) rectangle (0,1);
			\draw [fill = cyan] (-2,0) rectangle (-1,1);

			\node at (3.5, 0.5) {${\color{white}Y}$};
			\node at (2.5, 0.5) {${\mathbf p}$};
			\node at (1.0, 0.5) {${\sf IES}$};
			\node at (-0.5, 0.5) {${\mathbf p}$};
			\node at (-1.5, 0.5) {${\color{white}X}$};

			\draw [line width = 1, ->]  (4.2, 0.5)  to (5.8,0.5);
			\draw node at (5, 0.9) {{\sf cde}};

			\draw[fill = {red}](8,0) rectangle (9,1); 
			\draw[fill = {pink}](7,0) rectangle (8,1); 
			\draw[fill = {cyan}](6,0) rectangle (7,1); 
			
			\node at (8.5, 0.5) {${\color{white}Y}$};
			\node at (7.5, 0.5) {${\mathbf p}$};
			\node at (6.5, 0.5) {${\color{white}X}$};
			\end{tikzpicture}
\vspace{0.1in}
  \caption{Context Directed Excision: The IES flanked by pointer ${\mathbf p}$ on both sides is removed, along with one copy of ${\mathbf p}$.}
  \label{fig:cdediagram}
    \end{figure}
\end{center}

Observe that context directed block interchanges and context directed reversals do not decrease or increase the length of the string they operate on, and they retain all the pointers. 
But context directed excision, as illustrated in Figure \ref{fig:cdediagram}, changes the pointer contexts by deleting selected pointers and IESs.  

\section{Pointer lists}

Pointers are an essential ingredient of the three DNA editing operations. We exploit this central role of pointers by now basing our computational formalism (that mathematically models these three ciliate operations) on pointers. Towards this end we introduce the notion of a \emph{pointer list}\footnote{In anticipation of wider applicability of the notion of a pointer list we give a definition that is more general than the specific instance of it that we need.}.

\begin{definition}\label{def: PointerList} A finite sequence $P:=\lbrack x_1,\cdots,x_m \rbrack$ of integers is said to be a \emph{pointer list} if it satisfies the following six conditions:
\begin{enumerate}
  \item{$m$ is an even positive integer;} 
  \item{there is a unique $i$ with $\mu = \vert x_i\vert = \min\{\vert x_j\vert:1\le j\le m\}$;} 
  \item{there is a unique $j$ with $\lambda = \vert x_j\vert = \max\{\vert x_i\vert:1\le i\le m\}$;} 
  \item{For each $i \in \{1,\cdots, m\}$ with $\mu < \vert x_i \vert < \lambda$, there is a unique
        $j \in \{1,\cdots,n\}\setminus\{i\}$ such that $\vert x_i \vert = \vert x_j \vert$;} 
  \item{for each odd $i\in\{1,\cdots,n\}$, $x_i\le x_{i+1}$ and $x_i\cdot x_{i+1} > 0$;} 
  \item{whenever $i\in\{1,\,\cdots,\,n\}$ is odd, there is no $j$ such that $\vert x_i\vert<\vert x_j\vert< \vert x_{i+1}\vert$ or $\vert x_{i+1}\vert<\vert x_j\vert< \vert x_{i}\vert$.}
\end{enumerate}
\end{definition}

The following two mathematical facts are important in reasoning about ciliate operations on pointer lists. 
\begin{lemma}\label{appII:valueandparity} Let $\lbrack x_1,\, x_2,\,\cdots,x_{m-1},\,x_m\rbrack$ be a pointer list. If $i$ and $j$ be distinct indices for which $\vert x_i\vert = \vert x_j\vert$, then  $x_i$ and $x_j$ have the same sign if, and only if, $i$ and $j$ have distinct parity.
\end{lemma}

\begin{lemma}\label{appII:CentralLemma}
If $\lbrack x_1,\, x_2,\,\cdots,x_{m-1},\,x_m\rbrack$ is a pointer list of length larger than $4$, then at least one of the following three statements is false:
\begin{enumerate}
\item[(a)] $(\forall i)( x_i \neq x_{i+1})$
\item[(b)] $(\forall i)(\forall j)(\mbox{If } \vert x_i\vert = \vert x_j\vert, \mbox{ then } x_i = x_j)$
\item[(c)] $(\forall i)(\forall j)(\forall k)(\forall \ell)(\mbox{If }i \neq k,\, j\neq \ell,\, i < j \mbox{ and }x_i = x_k\mbox{ and }x_j = x_{\ell}, \mbox{ then either }i < j < \ell < k \mbox{ or }i < k < j < \ell)$
\end{enumerate}
\end{lemma}
In the interest of readability we postpone the somewhat lengthy, yet elementary, proofs of these facts to Appendix II. 

Pointer lists to which we will apply the ciliate operations come about as follows: Let $\integers$ denote the set of integers. For a set $S$ the symbol $^{<\omega}S$ denotes the set of finite sequences with entries from $S$. For an integer $z$ we define
\[
   \check{z}(1) =\begin{cases} z   & \mbox{if } z = \vert z\vert \\
                               z-1 & \mbox{otherwise } 
         \end{cases}
\] 
and  in all cases $\check{z}(2) = \check{z}(1)+1$. Then define the function $\pi:\,^{<\omega}\integers\rightarrow \,^{<\omega}\integers$ by:
\[
  \pi(\lbrack z_1,\cdots,z_k\rbrack) = \lbrack \check{z}_1(1),\check{z}_1(2),\cdots,\check{z}_k(1),\check{z}_k(2)\rbrack
\]
Thus, for example, $\pi(\lbrack -1,4,3,5,2,-9,7,10,-8,6\rbrack)$ is the sequence 
\[
  \lbrack -2,-1,4,5,3,4,5,6,2,3,-10,-9,7,8,10,11,-9,-8,6,7\rbrack.
\]
It can be verified that this sequence is indeed a pointer list.  The following lemma captures this fact.
\begin{lemma}\label{pointerassignment}
For each finite sequence $M:=\lbrack s_1,\, s_2,\, \cdots, \, s_n\rbrack$ of non-zero integers such that there is an integer $m$ for which $\{\vert s_i\vert:1\le i\le n\} = \{m+1,\, \cdots,\,m+n\}$, the sequence $\pi(M)$ is a pointer list.   
\end{lemma}
The proof consists of verifying that $\pi(M)$ meets all stipulations of Definition \ref{def: PointerList}.

\section{The ciliate operations on pointer lists}

We now introduce three special functions, {\sf cde}, {\sf cdr} and {\sf cds}, from $^{<\omega}\integers$ to $^{<\omega}\integers$, inspired by the three ciliate operations, as follows: For a given finite sequence $P:=\lbrack x_1,\cdots,x_m \rbrack$,\\

{\flushleft{\underline{\tt Context Directed Excision}:}} 
\[
  {\sf cde}(P) = \left\{
                     \begin{tabular}{ll}
                         $P$ & if there is no $i$ with $x_i = x_{i+1}$\\
                                &                                                          \\ 
                         $\lbrack x_1,\cdots,x_{i-1},x_{i+2},\cdots,x_m \rbrack$ & for $i$ mimimal with $x_i = x_{i+1}$, otherwise.
                     \end{tabular}
                \right.      
\]

{\flushleft{\underline{\tt Context Directed Reversal}:}} 
\[
  {\sf cdr}(P) = \left\{
                     \begin{tabular}{ll}
                         $P$ & if there are no $i<j$ \\
                             &  with $x_i = -x_j$  \\
                             &                             \\
                         $\lbrack x_1,\cdots,x_{i-1},x_i,{\color{red}\underline{-x_j,\cdots,-x_{i+1}}},\,x_{j+1},\cdots,x_m \rbrack$ & for the minimal $i$ with\\
                             &   $x_i = -x_j$, for a $j>i$
                     \end{tabular}
                \right.      
\]

{\flushleft{\underline{\tt Context Directed Block Swaps}:}} \\
${\sf cds}(P) = P$ if there are no $i<j<k<\ell$ with $x_i = x_k$ and $x_j = x_{\ell}$. However if there are $i<j<k<\ell$ with $x_i = x_k$ and $x_j = x_{\ell}$, then choose the least such $i$, and for it the least corresponding $j$, and define ${\sf cds}(P)$ to be 
\[
  \lbrack x_1,\cdots,x_i,{\color{blue}\underline{x_k,\cdots,x_{\ell}}},\,x_j,\cdots,x_{k-1},\,{\color{red}\underline{x_{i+1},\cdots,\,x_{j-1}}},\,x_{\ell+1},\,\cdots,\, x_m \rbrack
\]

These three operations have now been defined on arbitrary finite sequences of integers. They behave rather well on the subset ${\sf PL}=\{\sigma\in\,^{<\omega}\integers:\, \sigma \mbox{ is a pointer list}\}$ of their domain, as stated in the next two theorems. In the interest of readability we postpone their proofs to Appendix III. 

\begin{theorem}\label{appIII:domain} If $P$ is a pointer list of length larger than $4$, then at least one of the following statements is true:
\begin{enumerate}
  \item ${\sf cde}(P)\neq P$;
  \item ${\sf cdr}(P)\neq P$;
  \item ${\sf cds}(P)\neq P$.
\end{enumerate}
\end{theorem}

\begin{theorem}[Pointer list preservation]\label{appIII:parityth}
Let $P = \lbrack x_1,\cdots,x_m\rbrack$ be a pointer list.
Then each of ${\sf cde}(P)$, ${\sf cdr}(P)$ and ${\sf cds}(P)$ is a pointer list.
\end{theorem}

A finite sequence $\sigma$ is a \emph{fixed point} of a function $F:\,^{<\omega}\integers\rightarrow\,^{<\omega}\integers$ if $F(\sigma) = \sigma$. 

\begin{theorem}\label{appIII:aftercdrcds} If P is a pointer list of length larger than 4 and not a fixed point of $F\in\{{\sf cdr},\, {\sf cds}\}$, then $F$(P) is not a fixed point of {\sf cde}.
\end{theorem}

\section{The HNS algorithm}

Call a pointer list a \emph{destination} if it is one of the following: $\lbrack \mu,\lambda\rbrack$, $\lbrack -\lambda,-\mu\rbrack$, or for some integer $z$ with $\vert z\vert\not\in\{\lambda,\, \mu\}$, the pointer list is one of $\lbrack z,\,\lambda,\mu,\,z\rbrack$ or  $\lbrack z,\,-\mu,-\lambda,\,z\rbrack$. 

Let $P$ be a pointer list. Letting ${\sf cde}^i(P)$ denote the $i$-th iteration of ${\sf cde}$ on $P$, define $e(P)$ to be the minimal value of $i$ such that ${\sf cde}^{i+1}(P) = {\sf cde}^i(P)$. Then define ${\sf E}(P) = {\sf cde}^{e(P)}(P)$.

\begin{theorem}\label{algorithmtermination}  For a given pointer list $P_0$ define the sequence 
$
  P_0,\, P_1,\, \cdots,\, P_i,\, \cdots
$
so that
\[
  P_{i+1} = \left\{\begin{tabular}{ll}
                                    ${\sf E}(P_i)$ & if $P_i$ is not a {\sf cde} fixed point\\
                                                          &                                                             \\
                                    ${\sf cds}(P_i)$ & if $P_i$ is a ${\sf cde}$, but not a ${\sf cds}$ fixed point\\
                                                          &                                                             \\
                                    ${\sf cdr}(P_i)$ & if $P_i$ is a ${\sf cde}$ and a ${\sf cds}$ but not a ${\sf cdr}$ fixed point.     
                         \end{tabular}
                 \right.        
\]
Then the sequence 
$
  P_0,\, P_1,\, \cdots,\, P_i,\, \cdots
$
terminates in a destination.
\end{theorem}
\begin{proof} By Theorem \ref{appIII:parityth}, each term in this sequence is a pointer list. By Theorem \ref{appIII:domain}, as long as such a pointer list has more than four terms, it is not a fixed point of the ciliate operations. By Theorem \ref{appIII:aftercdrcds} the sequence does not terminate with an application of ${\sf cds}$ or of ${\sf cdr}$, but with an application of ${\sf E}$. Each application of $E$ reduces the length of a pointer list not a fixed point for ${\sf E}$ by a positive even number of terms. According to the definitions of the ciliate operations the pointers with absolute value $\lambda$ and $\mu$ are never excised, and thus present in any fixed point of a ciliate operation. Thus, a fixed point consisting of only two terms necessarily consists of the terms with absolute values $\lambda$ and $\mu$. As such a two term result is still a pointer list by Theorem \ref{appIII:parityth}, stipulation 5 of Definition \ref{def: PointerList} shows that this fixed point must be $\lbrack \mu,\, \lambda\rbrack$ or $\lbrack -\lambda,\, -\mu\rbrack$. Since applications of ${\sf cde}$ removes terms that are equal and adjacent, a four term fixed point must contain in addition to terms with absolute values $\mu$ and $\lambda$,  two terms of equal absolute value. If these two terms have opposite sign the pointer list is not a fixed point for ${\sf cdr}$. Thus, these two terms must be of the same sign.  But then, as the pointer list is a fixed point of ${\sf cde}$, these two terms are not adjacent. Moreover, their absolute value is strictly between $\mu$ and $\lambda$. Now stipulation 5 of Definition \ref{def: PointerList} implies that this pointer list is one of the two remaining claimed destinations.
\end{proof}

Thus the following algorithm, which we call the {\sf HNS} algorithm, halts:\\

{\tt
1) Input: A pointer list $P$, its length $\vert P\vert$ and integers $r$ and $s$;\\
2) Iteratively apply ${\sf cde}$ until a ${\sf cde}$ fixed point is reached. With each application, decrease $\vert P\vert$ by 2. Then proceed to 3).\\ 
3) If $P$ is a fixed point of {\sf cds}, proceed to 4). Else, apply {\sf cds}, increase $s$ by $1$, and return to 1).\\
4) If $P$ is a fixed point of {\sf cdr}, terminate the algorithm and report the current values of $P$, $r$ and $s$. Else, apply {\sf cdr}, increase $r$ by $1$, and return to 1). }\\

\tikzstyle{decision} = [diamond, draw, fill=blue!20, text width=2em, text badly centered, node distance=3cm] 
\tikzstyle{block} = [rectangle, draw, fill=blue!20, text width=6em, text centered, rounded corners, minimum height=4em, node distance=3cm]
\tikzstyle{block2} = [rectangle, draw, fill=green!20, text width=5em, text centered, rounded corners, minimum height=4em, node distance=3cm]
\tikzstyle{line} = [draw, -latex']
\tikzstyle{cloud} = [draw, ellipse,fill=red!20, node distance=3cm, minimum height=2em]
\begin{figure}    
\begin{tabular}{|c|}\hline
\\
  \begin{tikzpicture}[node distance = 2cm, auto]
    \node [cloud] (expert) {$P_0$, r=0, s=0};
    \node [block, below of=expert] (init) {Input: $P_i$\\ $\vert P_i\vert$, r, s};
    \node[left of=init]{Step 1};
    \node [decision, below of=init] (identify) {Is ${\sf cde}(P_i)=P_i$?};
    \node[left of=identify]{Step 2};
    \node [block, right of=identify] (excision) {Apply {\sf cde}.\\ i = i\\ $\vert P_i\vert = \vert P_i\vert-2$};
    \node [decision, below of=identify] (evaluate) {Is ${\sf cds}(P_i)=P_i$?};
    \node[left of=evaluate]{Step 3};
    \node [block, right of=evaluate, xshift=1.5cm ] (swap) {Apply {\sf cds}\\ $\vert P_{i+1}\vert = \vert P_i\vert$ \\ r = r+1, i = i+1};
    \node [decision, below of=evaluate] (reverse) {Is ${\sf cdr}(P_i)=P_i$?};
    \node[left of=reverse]{Step 4};
    \node [block, right of=reverse, xshift=3.0cm ] (invert) {Apply {\sf cdr}\\ $\vert P_{i+1}\vert = \vert P_i\vert$ \\ s = s+1, i=i+1};
    \node[block2, below of=reverse](output){Output: \\Destination,\\  r, s};

    \path [line] (init) -- (identify);
    \path [line] (identify) -- node[near start] {no} (excision);
    \path [line] (identify) -- node[near start] {yes} (evaluate);
    \path [line] (evaluate) -- node[near start] {no} (swap);
    \path [line] (evaluate) -- node[near start] {yes} (reverse);
    \path [line] (reverse) -- node[near start] {no} (invert);
   \path [line] (reverse) -- node[near start] {yes} (output);    \path[line,dashed] (excision) |- (init);
    \path[line,dashed] (swap) |- (init);
    \path[line,dashed] (invert) |- (init);
    \path [line,dashed] (expert) -- (init);
  \end{tikzpicture} \\
  \\ \hline
  \end{tabular}
  \caption{A flow diagram for the HNS algorithm.}\label{fig:HNSFlowDiagram}
  \end{figure}
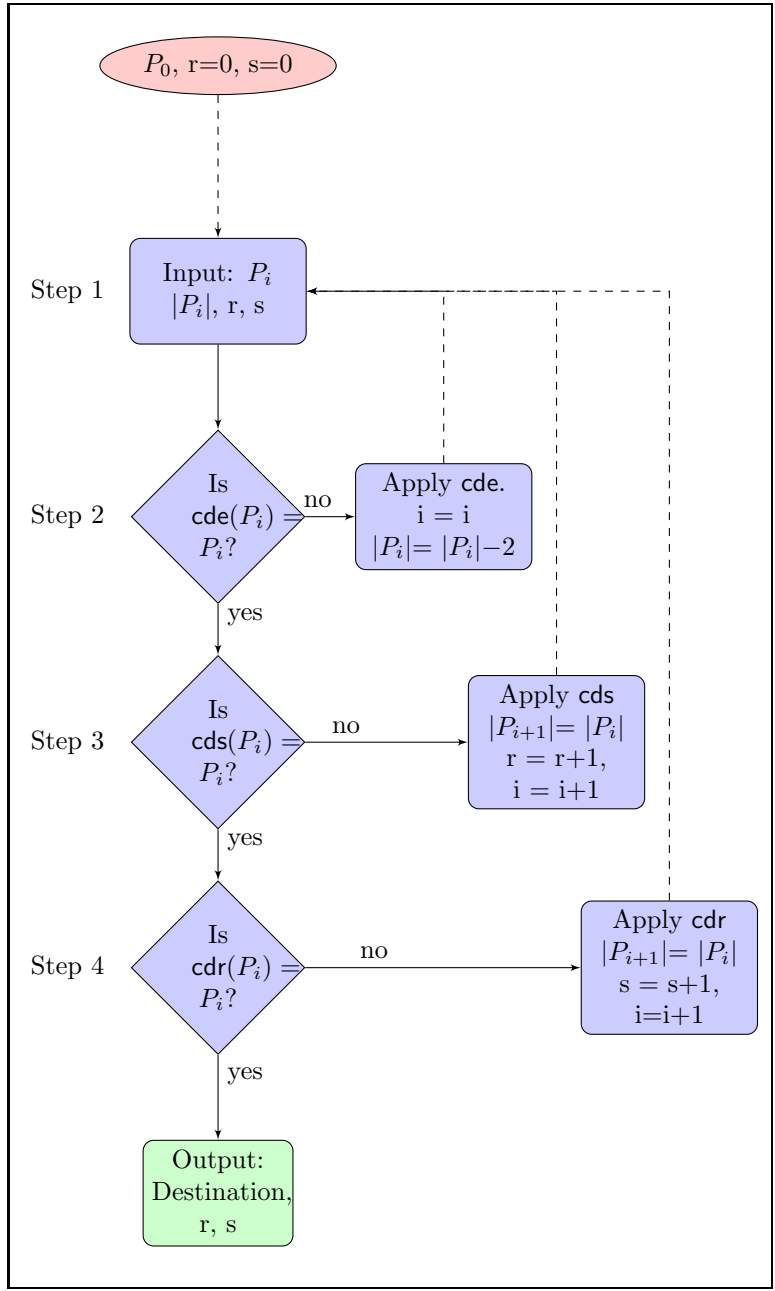

Figure \ref{fig:HNSFlowDiagram} depicts the algorithm in flow-diagram style. Let the original length of the pointer list $P$ be denoted $\vert P\vert$.\\ 
In step 2, the algorithm examines $\vert P\vert - 1$ adjacent pairs.  
If $P$ is not a {\sf cde} fixed point, then with the application of {\sf cde}, $\vert P\vert$ decreases by $2$. In this step we update the length of the resulting $P$ with each nontrivial application of {\sf cde}. \\ 
In step 3 the algorithm starts with a position $k<\vert P\vert$  
and then chooses a position $\ell>k+1$ with $x_k=x_{\ell}$ if any. 
This takes at most $(\vert P\vert - 1) + (\vert P\vert - 2) + \cdots + 2$ 
 search steps, which is $O(\vert P\vert^2)$. If this search fails, proceed to step 4. 
Else, suppose a successful $k+1<\ell<\vert P\vert$ 
 is found. Then for $k<j<\ell$ search for an $m>\ell$ with $x_m=x_{\ell}$. This would require at most $(\ell-k)*(\vert P\vert - \ell)$  
 steps. If this fails, proceed to step 4. Else, execute a {\sf cds} based on the found quadruple $(k,j,\ell,m)$, increase $s$ by 1, and return to step 1. Step 3 is completed in $O(\vert P\vert^2)$
  search steps.\\
In step 4 the algorithm starts with a position $k<\vert P\vert$  
and then scans positions $j>k$ until it finds an $x_j = -x_k$. The worst case scenario for this search is also $(\vert P\vert -1) + (\vert P\vert - 2) + \cdots + 2$, or $O(\vert P\vert^2)$. 
If the search succeeds, the result of {\sf cdr} is obtained in at most $\vert P\vert -1$ 
 search steps. Increase $r$ by 1, and return to step 2. Else, if the search fails, terminate the algorithm and report the current values of $P$, $r$ and $s$. 

In one cycle of executing steps until return to step 1, the worst case scenario employs at most $O(\vert P\vert^2)$ 
search and execution steps. For the next round an upper bound is $O((\vert P\vert  -1)^2) = O(\vert P\vert^2)$.  
This continues for at most $\frac{\vert P\vert}{2}$ 
rounds. Thus a global upper bound, in terms of the length of the initial pointer list, is $O(\vert P\vert^3)$.

The efficiency of this algorithm that produces from an initial pointer list a fixed point for the operations {\sf cde}, {\sf cds} and {\sf cdr} in  $O(\vert P\vert^3)$ 
steps can probably be 
improved. Additionally, this algorithm most likely does not minimize the number of steps taken, using {\sf cde}, {\sf cds} and {\sf cdr}, to reduce a pointer list to a fixed point. 

In our phylogenetic application below, any calibration of time span in terms of the number of operations required is based on the above {\sf HNS} algorithm as computational standard for the calibration.

\section{An application to genome phylogenetics.}

As illustrated in Figure \ref{fig:HumanMouseX}, for organisms $\text{S}_1$ and $\text{S}_2$ there may be synteny blocks of orthologous genes on corresponding 
chromosomes. Choose $S_1$ as reference and number the synteny blocks 
in their $5^{\prime}$ to 3$^{\prime}$ order of appearance on $S_1$'s chromosome as $1,\, 2,\, 3,\, \cdots,\, n$. In species $S_2$ the synteny blocks 
of these same genes may appear in a different order, and individual synteny blocks may also appear in orientation opposite from the orientation in $S_1$. Write the corresponding list of numbers in their order of appearance on $S_2$'s chromosome, making the number negative if the synteny block orientation is opposite to that in $S_1$. The result is a signed permutation of the list $1,\, 2,\, 3,\, \cdots,\, n$.

Now imagine that the list of synteny blocks for $S_1$ are the MDS's of a ciliate macro nuclear gene G, while the signed permutation that represents the corresponding list of synteny blocks for $S_2$  is the micro nuclear precursor of G. Take the number of operations the ciliate decryptome performs to convert the micro nuclear precursor to its macro nuclear version G as a measure of the evolutionary distance between the two chromosomes of $S_1$ and $S_2$. We used the {\sf HNS} algorithm to simulate the actions of the ciliate decryptome on the set of highly permuted genomes from various species of fruit flies.


The fruitfly genome is organized in
four\footnote{There are exceptions: See for example Figure 1 of \cite{SB}. None of the exceptional species is considered in our paper.} chromosomes, enumerated 1, 2, 3 and 4. These four chromosomes are traditionally divided into six so-called Muller elements. The left and right arms of chromosome 2 each is one of these Muller elements, and similarly for chromosome 3. Chromosome 1 is the X chromosome. The correspondence of chromosomal material to Muller elements is as follows:
\begin{center}
\begin{tabular}{||l|c:c:c:c:c:c|} \hline
Chromosome     & 1 = X & 2L & 2R  & 3L & 3R & 4 \\ \hline 
Muller Element   &     A   &  B  &   C  & D  &  E  & F \\ \hline
\end{tabular}
\end{center} 

The fruitfly genome has at least 13,600 confirmed genes (and counting), but is not expected to host significantly more genes. 
Recall that our definition of a ``synteny block" is more restrictive than the one used in \cite{BSRXSG}, where ``micro-inversions" are permitted. See for example Table 1 on p. 1662 of \cite{BSRXSG} for data on these more relaxed synteny blocks relative to the genome of \emph{D. melanogaster}.  
Between two species the number of synteny blocks can still be well over a thousand, as can be gleaned from Table 1 of \cite{BSRXSG}, where the more relaxed definition of ``synteny block" actually provides a lower bound on the number of synteny blocks as defined in our paper. 

According to findings of \cite{BSRXSG} 95\% of orthologous genes between two species are present on the same Muller element. For the species we are using, with one exception to be noted now, evidence suggests that all orthologous genes are present on the same Muller elements. Using data obtained from flybase.org we examined the permutation structure of these for the eight species \emph{D. melanogaster}, \emph{D. yakuba}, \emph{D. erecta}, \emph{D. sechellia}, \emph{D. mojavensis}, \emph{D. simulans}, \emph{D. grimshawi} and \emph{D. virilis}.  As illustrated in Figure 3 of \cite{BSRXSG} there is a translocation of genes between Muller elements B and C for \emph{D. erecta}, one of the species in our sample. Thus we combined Muller elements B and C into one computational unit (chromosome 2) for our application. Thus, we refer to the five units A, B/C, D, E and F in the remainder of this discussion.  

For each of the five units we computed, using in-house developed software written in Python, the number of applications of context directed swaps or context directed reversals performed by the {\sf HNS} algorithm to permute the gene order of one species to produce the corresponding gene order of another species. This was done with each species considered as reference species. Since {\sf HNS} gives preference to block interchanges the number of reversals in our derived data is low. 

Note that although we used the full gene lists from flybase.org, using pointer lists and ciliate operations automatically reduces to performing ciliate sorting operations on synteny blocks between pairs of species. 

From our data about the number of context directed swaps, $s$, and reversals, $r$, we define a corresponding distance matrix by using the formula $s+\frac{r}{2}$. As the reader would observe from examining our data, this in fact does define a metric\footnote{There are strong grounds for equating the value of two reversals with that of a single swap. As computations show, the result (given in Appendix I) is a matrix that is symmetric over its diagonal. It is also evident that the number of sorting operations to sort permutation $\alpha$ to obtain permutation $\beta$, plus the number of sorting operations to sort permutation $\beta$ to permutation $\gamma$, is no smaller than the number of sorting operations to directly sort permutation $\alpha$ to permutation $\gamma$. Thus, the triangle inequality holds.}

Then we applied the unweighted pair group method with arithmetic mean, also known as the UPGMA algorithm\footnote{This is algorithm 4.1 in \cite{CB}. A good exposition is also given in Chapter 27 of \cite{BBEGP}, available online at www.evolution-textbook.org.}, to these metrics. We used an in-house developed MAPLE implementation of UPGMA to compute these phylogenies. The corresponding phylogenetic trees were drawn using the ``newicktree" package for the LaTeX typesetting system.

Appendix I contains the data, derived distance matrices and corresponding phylogenetic trees for the five units in Figures \ref{fig:MullerA}, \ref{fig:MullerBC}, \ref{fig:MullerD}, \ref{fig:MullerE} and \ref{fig:MullerF}. An entry in the format ``r:s" in row i and column j of a table is interpreted as follows: ``r" denotes the number of context directed reversals ({\sf cdr} operations), while ``s" denotes the number of context directed block interchanges ({\sf cds} operations) executed by the {\sf HNS} algorithm to convert the permutation of the species in row i  to that of the species in column j. Thus the species in column j is the \emph{reference species}. The total for whole genomes is given in Figure \ref{fig:MullerWG}. 

We used the timeline given in figures 1 and 3 of \cite{HHH} to calibrate the time line in our phylogenetic trees\footnote{We could have used alternative timelines, such as for example the timelines given in the figure at the DroSpeGe web site {\tt http://insects.eugenes.org/DroSpeGe/}. Whichever published timeline one chooses will determine the corresponding calibration applied to our data.} 
This calibration is a rough time line: Our work describes evolutionary relationships among instances of a specific chromosome present in these eight species. The evolutionary time line for a chromosome need not agree with the evolutionary time line for speciation. According to Figures 1 and 3 of \cite{HHH} the time span from the earliest common ancestor of our species is roughly 60 million years.
\section*{Discussion}

Comparison of our results in Appendix I, and the results of \cite{BSRXSG}  Table 2, show a significant difference in the number of sorting operations, with ours typically higher. One reason for these differences lies in our definition of synteny blocks: We allow blocks consisting of a single gene, and we do not allow blocks containing different gene orders. Thus, we have a larger number of synteny blocks to be sorted, and our computations took into account all orthologous genes. This point is illustrated by comparing the number of synteny blocks for Muller element E for \emph{D. yakuba}, \emph{D. sechellia} and \emph{D. simulans} (computed relative to \emph{D. melanogaster}) reported in Table 5 of \cite{BSRXSG} with the actual number of sorting operations reported for these species (with \emph{D. melanogaster} as reference) in our Figure \ref{fig:MullerE}. 
Moreover, whereas in \cite{BSRXSG} the authors used unconstrained reversals as sorting operation, we used context directed reversals. Additionally, in \cite{BSRXSG} genes that suggest that a transposition is responsible for the rearrangement were excluded from the analysis. We included all orthologous genes since the sorting operation of context directed swaps (block interchanges) accounts also for transpositions.

Comparison of the phylogenies in Appendix I with the phylogeny in Figure 8 of \cite{BSRXSG} or with the phylogeny of sequenced species at flybase.org\footnote{{\tt http://flybase.org/static\_pages/species/sequenced_species.html}} indicate that our placement of \emph{D. sechellia} is in all cases quite different. The placement of \emph{D. mojavensis, D. virilis} and \emph{D. grimshawi} relative to each other and to the other species agrees with both of these phylogenies for all but Muller elements A and E.  

By using the UPGMA algorithm to construct phylogenies from distance matrices we assumed a uniform rate of evolution for the Muller elements. Comparing these uniform rates among the different chromosomes indicate that no two individual chromosomes undergo permutations at the same rates. Our sorting data suggests the upper bounds in Figure \ref{fig:PermRates} on the number of ciliate sorting operations (cso) since the most recent common ancestor of all the species considered. 

\begin{figure}[h]
\begin{center}
\begin{tabular}{|l|r|} \hline
Computational unit  & cso \\ \hline
A             &  266.75 \\
B/C         &  207.75 \\
D            &  247.25 \\ 
E            &  364.25 \\ 
F            &       8.5  \\ \hline
\end{tabular}\\
\caption{Ciliate sorting operations since most recent common ancestor of all considered species} 
\label{fig:PermRates}
\end{center}
\end{figure}
These numbers were computed by taking the largest ciliate sorting distance achieved between a pair of the considered species, and dividing\footnote{Using our hypothesis of uniform rate of evolution} by 2 to obtain an estimate of the number of ciliate sorting operations to each species' corresponding genomic element since their most recent common ancestor. 

The Muller F element has undergone remarkably few permutations in comparison with the other Muller elements.  Muller element E appears to be the most susceptible to permutation, while Muller element F appears the most ``resistant" to permutation. This, however, may be a biased view of susceptibility to permutation since these computational units do not harbor the same number of genes or synteny blocks. As indicated in \cite{Hochman}, Chromosome 4 (Muller element F) is generally a very small chromosome: it may contain fewer than 100 genes only (see for example the results regarding Muller element F for various species in \cite{SB}). The other Muller elements each contains well over 1000 genes each. Thus one would expect the number of rearrangements needed to sort one species' Chromosome 4 gene content to that of another species to be relatively low in comparison with the other, larger, chromosomes.

Tables 5 and 6 of \cite{BSRXSG} report rearrangement rates that are computed from the number of synteny blocks relative to \emph{D. melanogaster}, the nucleotide length of the Muller element, and the estimated divergence time for the species in question. These rates assume that arbitrary reversals cause the rearrangements and thus ignore genes deemed to have been moved by other sorting mechanisms, and use a definition of ``synteny block" that ignores certain rearrangements. In the case of our context directed sorting operations a more appropriate measure of ``susceptibility to permutation" should probably take into account additional parameters regarding nucleotide patterns in the Muller elements. Progress in this regard would address the third\footnote{``$\dots$ how do new inversions originate?" This can be expanded to include the question of how new block interchanges originate.} and fourth\footnote{``$\dots$ what is the molecular basis for gene arrangement polymorphism?"} questions raised on p. 1603 and 1604 of \cite{SB}, phrased for arbitrary reversals, and may also indicate whether context directed reversals and block interchanges are more suitable sorting operations for phylogenetic analyses based on permutations of genomic material. Such rearrangement rates may be used as ``susceptibility coefficients", measuring the susceptibility of a genomic element to rearrangement.


According to Figure 3 of \cite{BSRXSG} the F element of \emph{D. willistoni} (which is not among the species we considered) has been absorbed in the E-element of \emph{D. willistoni}. It would be interesting to ``distill" the \emph{D. willistoni} F-element from the \emph{D. willistoni} E-element, and compare its level of permutation relative to the F-element of the eight species in our study. Establishing susceptibility coefficients may enable us to obtain from the current permutation state of the distilled ``\emph{D. willistoni} F-element", and an established evolutionary time distances for the fruitfly phylogeny, an estimate of when absorption of the F-element into the E-element took place.

Similarly, by separating the treatment of the B and C elements, and calculating the corresponding susceptibility coefficients of these elements, and distilling the B-element components and the C-element components for \emph{D. ananassae}, one may be able to estimate when these transpositions occurred. Figure 3 of \cite{BSRXSG} also indicates that part of \emph{D. pseudoobscura}'s Muller A element was transposed to its Muller E element. Susceptibility coefficients may be useful in estimating when this transposition occurred. An investigation of the structural properties of the chromosomes involved in these inter chromosomal translocations may also reveal if any DNA motifs promote these translocations.

The differences in phylogenies for different chromosomal domains in the considered species suggest the possibility of inferring from Mendelian inheritance hypotheses and diploidy of the fruitfly genomes, inter breeding among ancestor species that would produce the observed chromosomal configurations. 

We relied on the UPGMA algorithm for constructing our phylogenies. Other clustering techniques such as Neighborhood joining, or several other algorithms as for example in \cite{CB}, may reveal finer details than the technique applied here.


While using ciliate operations to compute the permutation based distances between pairs of species we found permutations which are not reducible to each other by ciliate operations. In contrast to the case for unrestricted block interchanges and unrestricted reversals, not all permutations are invertible by context directed block interchanges and reversals. When our algorithm terminates with a destination of length 4 instead of 2, this indicates that the two permutations involved in the distance measure requires an additional transposition to complete the transformation. Though we have not done so in our current paper, the fact of uninvertibility by ciliate decryptome operations could be taken as an additional parameter in measuring evolutionary distance. Instead, in this paper we counted this additional transposition needed at the end as a single step towards the distance. An argument can be made that the necessity of this additional transposition should be accounted for more significantly in computing evolutionary distance. It also raises the question of determining an easily applicable characterization of permutations that are invertible by constrained block interchanges or reversals. The problem of mathematically characterizing permutations that are invertible by context directed operations has been solved in subsequent work \cite{AHMMS}.

Finally, although the HNS algorithm finds in polynomial time the data needed to construct a distance matrix, we do not propose that this algorithm finds optimal data in the following sense: When one permutation can be transformed to another by means of context directed reversals and block interchanges, what is the least number of these operations needed for such a transformation? The answer for context directed block interchanges has been obtained in \cite{AHMMS}. The minimal number of operations may depend on strategic sorting decisions made while sorting a permutation. One may inquire whether certain permutations require less strategic decision making in order to obtain a successful sorting. The permutations requiring the least number of strategic decisions for context directed block interchanges have been characterized in \cite{ASWW}, but a complete answer is currently not known.

\vspace{-1in}



\newpage

\section*{Appendix I: The distance matrices underlying the application of UPGMA to the five chromosomes of eight fruitfly species.}

\begin{figure}[h] 
  \centering
\begin{tabular}{|c||c|c|c|c|c|c|c|c|}\hline
      & D. vir & D. gri & D. sim   & D. moj. & D. mel   & D.ere   & D.yak.  & D.sec. \\ \hline \hline
D.vir &        & 32:431 &  38:463  &  31:438 &   33:403 &  40:426 &  29:414 & 35:514 \\ \hline
D.gri & 26:434 &        &  36:446  &  35:430 &   36:381 &  45:404 &  40:391 & 45:504 \\ \hline
D.sim & 36:464 & 34:447 &          &  35:460 &    8:268 &  19:311 &  21:282 & 26:505 \\ \hline
D.moj & 29:439 & 37:429 &  41:457  &         &   41:407 &  34:434 &  36:422 & 37:515 \\ \hline
D.mel & 37:401 & 40:379 &   6:269  &  45:405 &          &   1:171 &  35:93  & 19:482 \\ \hline
D.ere & 36:428 & 43:405 &  29:306  &  42:430 &    3:170 &         &  25:182 & 28:499 \\ \hline
D.yak & 43:407 & 40:391 &  31:277  &  50:415 &   11:105 &  25:182 &         & 29:481 \\ \hline
D.sec & 43:510 & 39:507 &  20:508  &  39:514 &    7:488 &  22:502 &  17:487 &        \\ \hline
\end{tabular}
\begin{center}{A: r:s denotes number of cdr: number of cds}\end{center}
\vspace{0.1in}

\begin{tabular}{|c||c|c|c|c|c|c|c|c|}\hline
      &  D. vir & D. gri & D. sim   & D. moj. & D. mel   & D.ere   & D.yak.  & D.sec. \\ \hline \hline
D.vir &  0.0    & 447.0  & 482.0    & 453.5   & 419.5    & 446.0   & 428.5   & 531.5\\ \hline
D.gri &447.0    &   0.0  & 464.0    & 447.5   & 399.0    & 426.5   & 411.0   & 526.5\\ \hline
D.sim &482.0    & 464.0  & 0.0      & 477.5   & 272.0    & 320.5   & 292.5   & 518.0\\ \hline
D.moj &453.5    & 447.5  & 477.5    & 0.0     & 427.5    & 451.0   & 440.0   & 533.5\\ \hline
D.mel &419.5    & 399.0  & 272.0    & 427.5   & 0.0      & 171.5   & 110.5   & 491.5\\ \hline
D.ere &446.0    & 426.5  & 320.5    & 451.0   & 171.5    & 0.0     & 194.5   & 513.0\\ \hline
D.yak &428.5    & 411.0  & 292.5    & 440.0   & 110.5    & 194.5   & 0.0     & 495.5\\ \hline
D.sec &531.5    & 526.5  & 518.0    & 533.5   & 491.5    & 513.0   & 495.5   & 0.0\\ \hline
\end{tabular}
\begin{center}{Distance matrix for Muller Element A}\end{center}
\vspace{0.1in}

  \centering
  \includegraphics[bb=0 0 488 412,width=2.56in,height=2.16in,keepaspectratio]{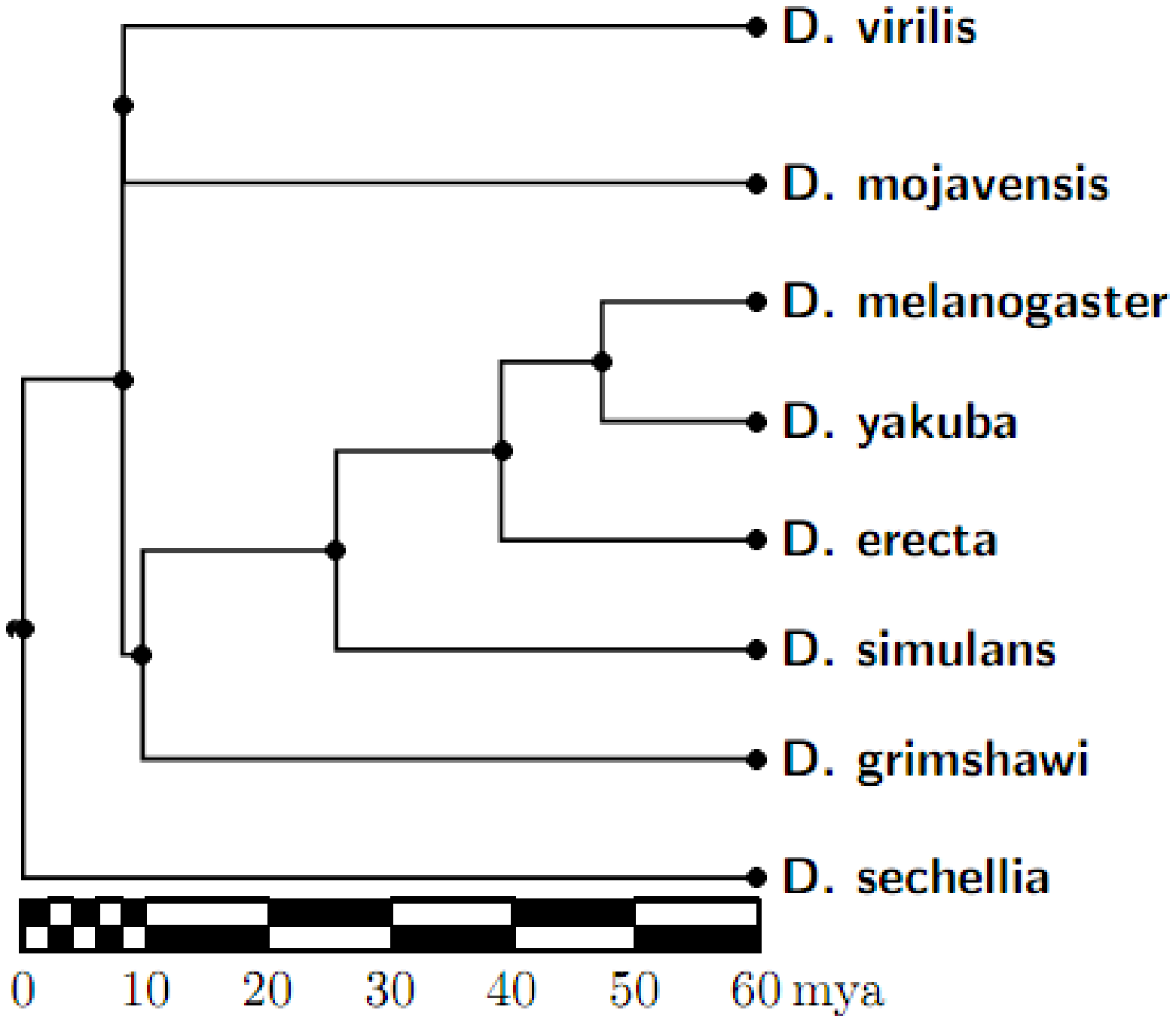}
\begin{center}{Phylogeny for the Muller A element}\end{center}
\vspace{0.1in}
  \caption{Data, distance matrix and resulting phylogeny for the Muller A-element} 
  \label{fig:MullerA} 
\end{figure}

\pagebreak

\begin{figure}[ht] 
  \centering
\begin{tabular}{|c||c|c|c|c|c|c|c|c|}\hline
      & D. vir & D. gri & D. sim   & D. moj. & D. mel   & D.ere   & D.yak.  & D.sec. \\ \hline \hline
D.vir &        & 62:255 &  44:290  &  31:161 &   41:262 &  51:324 &  63:332 & 38:375 \\ \hline
D.gri & 56:258 &        &  43:318  &  38:187 &   52:280 &  62:342 &  50:370 & 45:393 \\ \hline
D.sim & 48:288 & 49:315 &          &  53:205 &    2: 95 &  24:188 &  16:223 &  9:256 \\ \hline
D.moj & 67:143 & 32:190 &  55:204  &         &   49:173 &  51:254 &  48:285 & 47:319 \\ \hline
D.mel & 59:253 & 58:277 &   8: 92  &  49:173 &          &  19:159 & 101:145 &  3:229 \\ \hline
D.ere & 45:327 & 44:351 &  14:193  &  57:251 &   11:163 &         &   9:249 & 32:275 \\ \hline
D.yak & 49:339 & 42:374 &  14:224  &  44:287 &    7:192 &  15:246 &         & 34:286 \\ \hline
D.sec & 52:368 & 49:391 &   7:257  &  41:322 &    3:229 &  38:272 &  46:280 &        \\ \hline
\end{tabular}
\begin{center}{B/C: r:s denotes number of cdr: number of cds}\end{center}
\vspace{0.1in}

  \centering
\begin{tabular}{|c||c|c|c|c|c|c|c|c|}\hline
      &D. vir & D. gri & D. sim   & D. moj. & D. mel   & D.ere   & D.yak.  & D.sec. \\ \hline \hline
D.vir &       & 286    &  312     &  176.5  &   282.5  &  349.5  &  363.5  &  394 \\ \hline
D.gri & 286   &        &  339.5   &  206    &   306    &  373    &  395    &  415.5 \\ \hline
D.sim & 312   & 339.5  &          &  231.5  &   96     &  200    &  331    &  260.5 \\ \hline
D.moj & 176.5 & 206    &  231.5   &         &   197.5  &  279.5  &  309    &  342.5 \\ \hline
D.mel & 282.5 & 306    &   96     &  197.5  &          &  168.5  &  195.5  &  230.5 \\ \hline
D.ere & 349.5 & 373    &  200     &  279.5  &   168.5  &         &  253.5  &  291 \\ \hline
D.yak & 363.5 & 395    &  331     &  309    &   195.5  &  253.5  &         &  303 \\ \hline
D.sec & 394   & 415.5  &  260.5   &  342.5  &   230.5  &  291    &  303    &        \\ \hline
\end{tabular}
\begin{center}{Distance matrix for Muller Element B/C}\end{center}
\vspace{0.1in}

  \centering
  \includegraphics[bb=0 0 488 412,width=2.56in,height=2.16in,keepaspectratio]{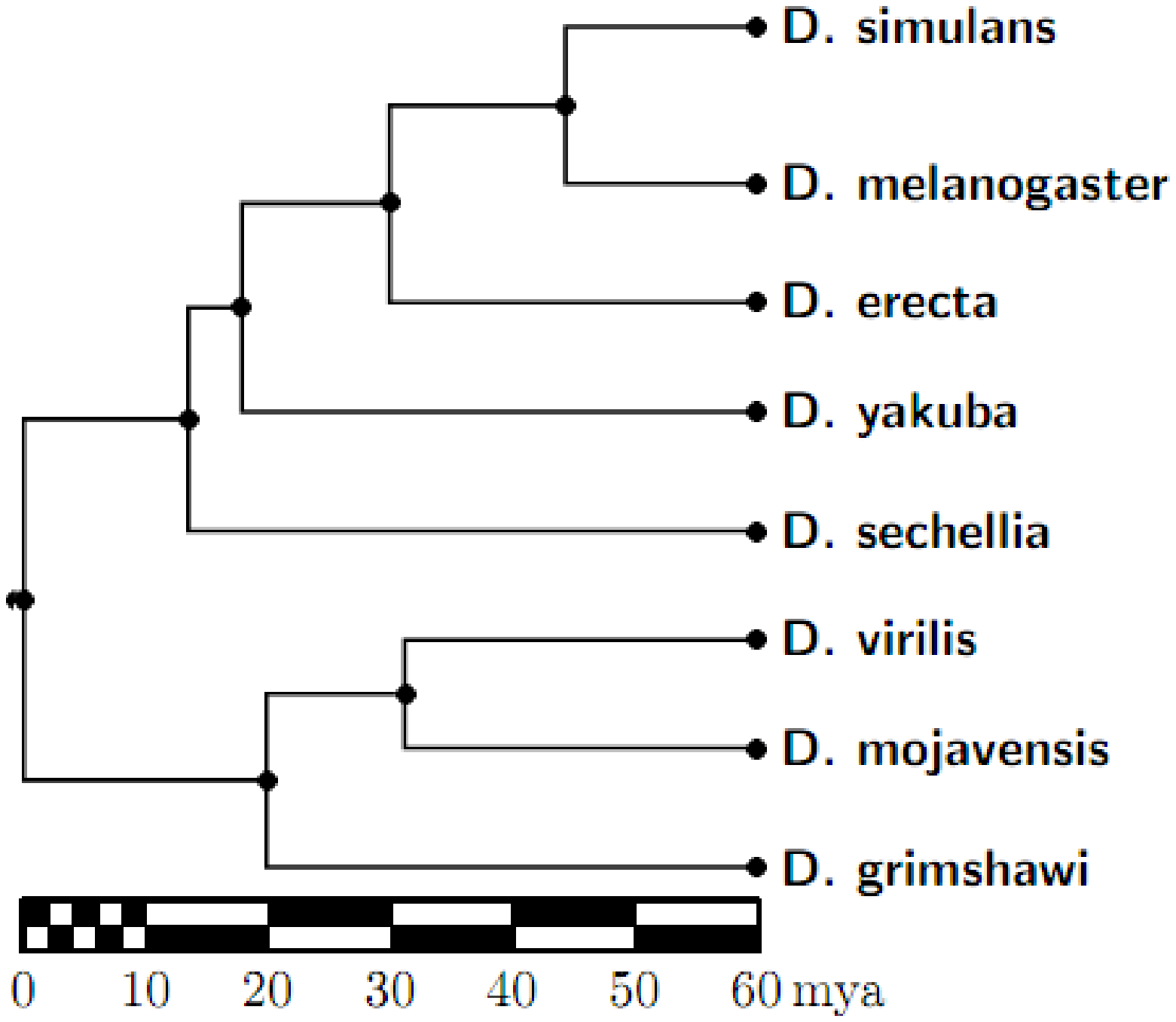}
\begin{center}{Phylogeny for the Muller B/C element}\end{center}
\vspace{0.1in}
  \caption{Data, distance matrix and resulting phylogeny for the Muller B/C-element} 
  \label{fig:MullerBC} 
\end{figure}

\pagebreak

\begin{figure}[ht]
  \centering
\begin{tabular}{|c||c|c|c|c|c|c|c|c|}\hline
      & D. vir & D. gri & D. sim   & D. moj. & D. mel   & D.ere   & D.yak.  & D.sec. \\ \hline \hline
D.vir &        & 21:124 &  60:193  &  27:113 &   69:175 &  58:160 &  53:231 & 60:450 \\ \hline
D.gri & 27:121 &        &  51:210  &  29:154 &   52:187 &  56:174 &  56:244 & 59:460 \\ \hline
D.sim & 68:189 & 59:206 &          &  65:219 &    2: 69 &   5: 56 &  10:129 &  2:390 \\ \hline
D.moj & 23:115 & 23:157 &  59:222  &         &   53:214 &  59:192 &  55:257 & 51:469 \\ \hline
D.mel & 69:175 & 62:182 &   2: 69  &  81:200 &          &   8: 35 &  10:109 &  0:388 \\ \hline
D.ere & 66:156 & 58:173 &   7: 55  &  67:188 &   10: 34 &         &  12: 79 & 90:337 \\ \hline
D.yak & 59:228 & 64:240 &  26:121  &  71:249 &   14:107 &  18: 76 &         & 12:416 \\ \hline
D.sec & 54:453 & 55:462 &   2:390  &  49:470 &    0:388 &   4:380 &  12:416 &        \\ \hline
\end{tabular}
\begin{center}{D: r:s denotes number of cdr: number of cds}\end{center}
\vspace{0.1in}

  \centering
\begin{tabular}{|c||c|c|c|c|c|c|c|c|}\hline
      & D. vir & D. gri & D. sim   & D. moj. & D. mel   & D.ere   & D.yak.  & D.sec. \\ \hline \hline
D.vir &   0    & 134.5  &  223     &  126.5  &   209.5  &  189    &  257.5  & 480    \\ \hline
D.gri & 134.5  &   0    &  235.5   &  168.5  &   213    &  202    &  272    & 489.5  \\ \hline
D.sim & 223    & 235.5  &    0     &  251.5  &    70    &   58.5  &  134    & 391    \\ \hline
D.moj & 126.5  & 168.5  &  251.5   &    0    &   240.5  &  221.5  &  284.5  & 494.5  \\ \hline
D.mel & 209.5  & 213    &   70     &  240.5  &     0    &   39    &  114    & 388    \\ \hline
D.ere & 189    & 202    &   58.5   &  221.5  &    39    &    0    &   85    & 382    \\ \hline
D.yak & 257.5  & 272    &   134    &  284.5 &    114    &   85    &    0    & 422    \\ \hline
D.sec & 480    & 489.5  &   391    &  494.5  &   388    &  382    &  422    &   0    \\ \hline
\end{tabular}
\begin{center}{Distance matrix for Muller Element D}\end{center}
\vspace{0.1in}

  \centering
  \includegraphics[bb=0 0 488 412,width=2.56in,height=2.16in,keepaspectratio]{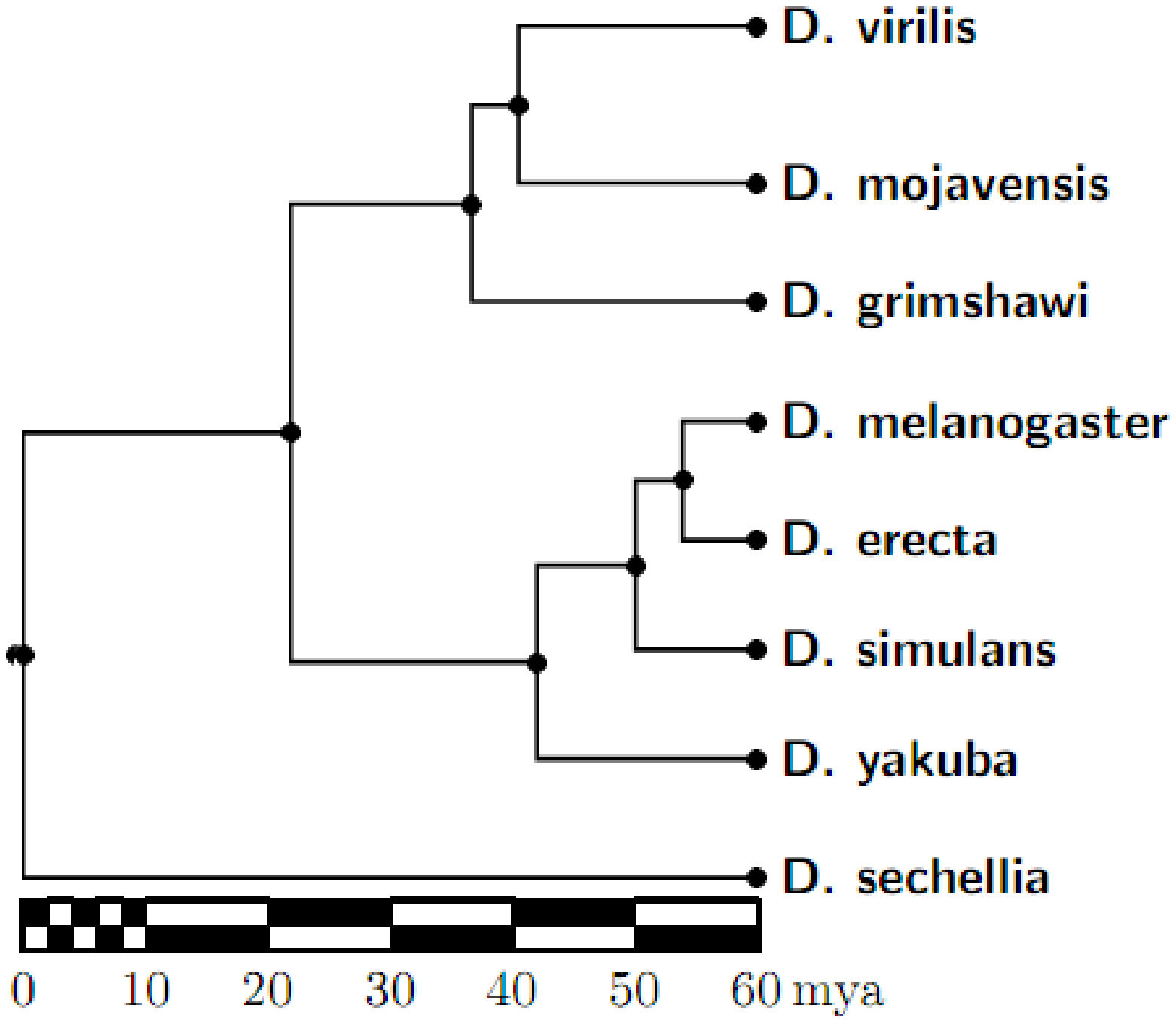}
\begin{center}{Phylogeny for the Muller D element}\end{center}
\vspace{0.1in}
  \caption{Data, distance matrix and resulting phylogeny for the Muller D-element} 
  \label{fig:MullerD} 
\end{figure}

\pagebreak

\begin{figure}[h] 
  \centering
\begin{tabular}{|c||c|c|c|c|c|c|c|c|}\hline
      & D. vir & D. gri & D. sim   & D. moj. & D. mel   & D.ere   & D.yak.  & D.sec. \\ \hline \hline
D.vir &        & 47:634 &  40:451  &  27:340 &   40:432 &  46:551 &  42:436 & 41:598 \\ \hline
D.gri & 47:634 &        &  47:616  &  25:549 &   46:602 &  55:664 &  54:603 & 47:705 \\ \hline
D.sim & 52:445 & 57:611 &          &  45:213 &    8: 71 & 142:241 &  13: 75 & 14:347 \\ \hline
D.moj & 89:309 & 53:535 &  39:216  &         &   39:185 &  43:401 &  31:194 & 45:446 \\ \hline
D.mel & 44:430 & 54:598 &   8: 71  &  39:185 &          & 196:196 &   7: 38 & 21:334 \\ \hline
D.ere & 50:549 & 55:664 &  10:307  &  39:403 &    6:291 &         &  12:291 & 38:428 \\ \hline
D.yak & 54:430 & 62:599 &  19: 72  &  43:188 &   15: 34 &   8:293 &         & 23:334 \\ \hline
D.sec & 51:593 & 53:702 &  14:347  &  39:449 &    5:342 &  38:428 &   9:341 &        \\ \hline
\end{tabular}
\begin{center}{E: r:s denotes number of cdr: number of cds}\end{center}
\vspace{0.1in}

  \centering
\begin{tabular}{|c||c|c|c|c|c|c|c|c|}\hline
      & D. vir & D. gri & D. sim   & D. moj. & D. mel   & D.ere   & D.yak.  & D.sec. \\ \hline \hline
D.vir &        & 657.5  &  471     &  353.5  &   452    &  574    &  457    & 618.5  \\ \hline
D.gri & 657.5  &        &  639.5   &  561.5  &   625    &  691.5  &  630    & 728.5  \\ \hline
D.sim & 471    & 639.5  &          &  235.5  &    75    &  312    &  81.5   & 354    \\ \hline
D.moj & 353.5  & 561.5  &  235.5   &         &   204.5  &  422.5  &  209.5  & 468.5  \\ \hline
D.mel & 452    & 625    &   75     &  204.5  &          &  294    &  41.5   & 344.5  \\ \hline
D.ere & 574    & 691.5  &  312     &  422.5  &    294   &         &  297    & 447    \\ \hline
D.yak & 457    & 630    &  81.5    &  209.5  &    41.5  &  297    &         & 345.5  \\ \hline
D.sec & 618.5  & 728.5  &  354     &  468.5  &    344.5 &  447    &  345.5  &        \\ \hline
\end{tabular}
\begin{center}{Distance matrix for Muller Element E}\end{center}
\vspace{0.1in}

  \centering
  \includegraphics[bb=0 0 488 412,width=2.56in,height=2.16in,keepaspectratio]{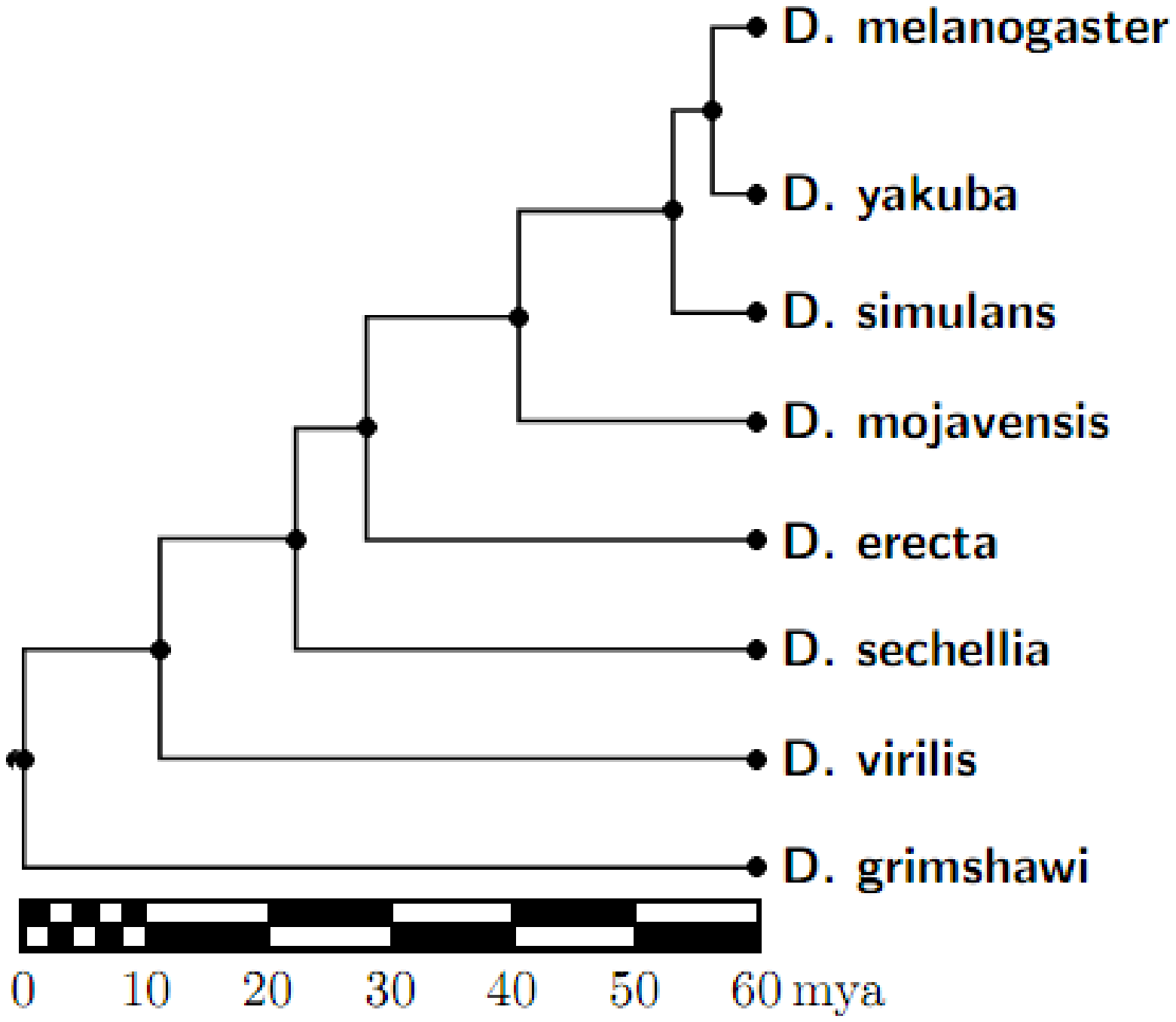}
\vspace{0.1in}
  \caption{Data, distance matrix and resulting phylogeny for the Muller E-element} 
  \label{fig:MullerE} 
\end{figure}

\pagebreak

\begin{figure}[h] 
  \centering
\begin{tabular}{|c||c|c|c|c|c|c|c|c|}\hline
      & D. vir & D. gri & D. sim      & D. moj.     & D. mel      & D.ere       & D.yak.      & D.sec.\\ \hline \hline
D.vir &        & 3:5    &   12:8      &   2:1       &   11:6      &  11:6       &  11:6       &  8:13 \\ \hline
D.gri & 3:5    &        &   10:12     &   3:4       &   11:9      &  11:9       &  11:9       & 10:12 \\ \hline
D.sim & 8:10   & 8:13   &             &  10:9       &    4:5      &   4:5       &   4:5       &  4:13 \\ \hline
D.moj & 2:1    & 3:4    &   6:11      &             &    7:7      &   7:7       &   7:7       &  9:12 \\ \hline
D.mel & 9:7    & 7:11   &   6:4       &   7:7       &             &   0:0       &   0:0       &  0:12 \\ \hline
D.ere & 9:7    & 11:9   &   6:4       &   9:6       &    0:0      &             &   0:0       &  0:12 \\ \hline
D.yak & 9:7    & 7:11   &   6:4       &   7:7       &    0:0      &   0:0       &             &  0:12 \\ \hline
D.sec & 8:13   & 8:13   &   2:14      &   9:12      &    0:12     &   0:12      &   0:12      &       \\ \hline
\end{tabular}
\begin{center}{F: r:s denotes number of cdr: number of cds}\end{center}
\vspace{0.1in}

  \centering
\begin{tabular}{|c||c|c|c|c|c|c|c|c|}\hline
      & D. vir & D. gri & D. sim      & D. moj.     & D. mel      & D.ere       & D.yak.      & D.sec.\\ \hline \hline
D.vir &        & 6.5    &   14      &    2         &   11.5      &  11.5       &  11.5       &  17 \\ \hline
D.gri &  6.5   &        &   17      &    5.5       &   14.5      &  14.5       &  14.5       &  17 \\ \hline
D.sim & 14     & 17     &           &   14         &    7        &   7         &   7         &  15 \\ \hline
D.moj &  2     & 5.5    &   14      &              &    10.5     &  10.5       &  10.5       &  16.5 \\ \hline
D.mel & 11.5   & 14.5   &   7       &   10.5       &             &     0       &     0       &  12 \\ \hline
D.ere & 11.5   & 14.5   &   7       &   10.5       &     0       &             &     0       &  12 \\ \hline
D.yak & 11.5   & 14.5   &   7       &   10.5       &     0       &     0       &             &  12 \\ \hline
D.sec & 17     & 17     &   15      &   16.5      &      12      &     12      &     12      &       \\ \hline
\end{tabular}
\begin{center}{Distance matrix for Muller Element F}\end{center}
\vspace{0.1in}

  \centering
  \includegraphics[bb=0 0 488 412,width=2.56in,height=2.16in,keepaspectratio]{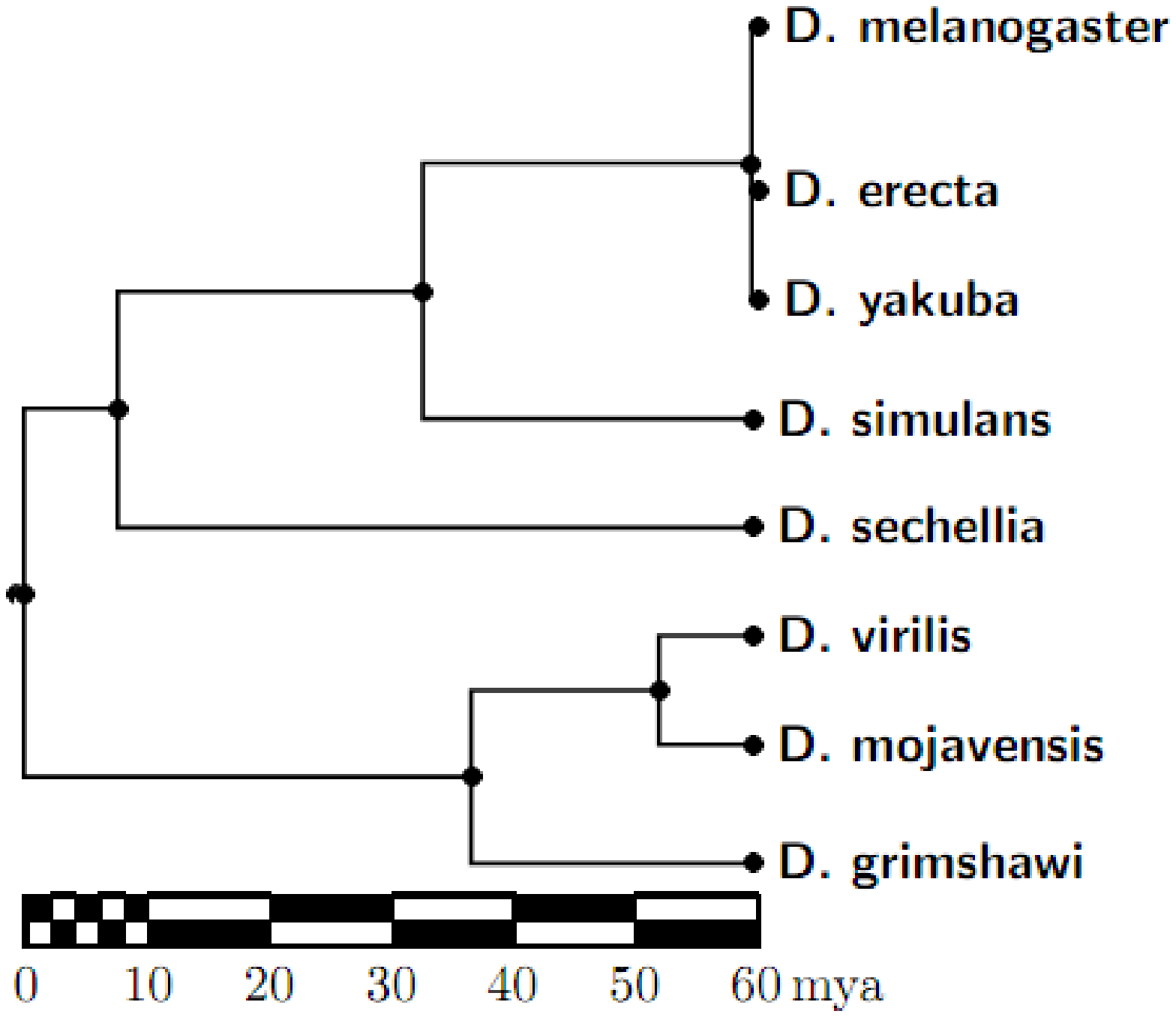}
\vspace{0.1in}
  \caption{Data, distance matrix and resulting phylogeny for the Muller F-element} 
  \label{fig:MullerF} 
\end{figure}

\pagebreak

\begin{figure}[h] 
  \centering
\begin{tabular}{|c||c|c|c|c|c|c|c|c|}\hline
      & D. vir  & D. gri  & D. sim    & D. moj.    & D. mel   & D.ere    & D.yak.     & D.sec.    \\ \hline \hline
D.vir &         & 165:1449&  194:1405 &   118:1053 &  194:1278&  206:1467&  198:1419  &  182:1950 \\ \hline
D.gri & 165:1449&         &  187:1602 &   130:1324 &  197:1459&  229:1593&  211:1617  &  206:2074 \\ \hline
D.sim & 194:1405& 187:1602&           &   208:1106 &   24:508 &  194:801 &   64:722   &   55:1511 \\ \hline
D.moj & 118:1053& 130:1324&  208:1106 &            &  189:986 &  194:1288&  177:1165  &  189:1761 \\ \hline
D.mel & 194:1278& 197:1459&   24:508  &   189:986  &          &  224:561 &  153:385   &   43:1445 \\ \hline
D.ere & 206:1467& 229:1593&  194:801  &   194:1288 &  224:561 &          &   58:801   &  188:1551 \\ \hline
D.yak & 198:1419& 211:1617&   64:722  &   177:1165 &  153:385 &   58:801 &            &   98:1529 \\ \hline
D.sec & 182:1950& 206:2074&  55:1511  &   189:1761 &   43:1445&  188:1551&   98:1529  &           \\ \hline
\end{tabular}
\begin{center}{Whole genome: r:s denotes number of cdr: number of cds}\end{center}
\vspace{0.1in}

  \centering
\begin{tabular}{|c||c|c|c|c|c|c|c|c|}\hline
      & D. vir & D. gri & D. sim      & D. moj.        & D. mel  & D.ere  & D.yak.  & D.sec.\\ \hline \hline
D.vir &    0   & 1231.5 &   1502      &    1112        &   1375  &  1570  &  1518   &  2041 \\ \hline
D.gri & 1231.5 &    0   &  1695.5     &    1389        &   1557.5&  1707.5&  1722.5 &  2177 \\ \hline
D.sim & 1502   & 1695.5 &    0        &    1210        &    520  &   898  &   746   &  1538.5 \\ \hline
D.moj & 1112   & 1389   &  1210       &    0           &   1080.5&  1385  &  1253.5 &  1655.5 \\ \hline
D.mel & 1375   & 1557.5 &   520       &   1080.5       &     0   &   673  &   461.5 &  1463.5 \\ \hline
D.ere & 1570   & 1707.5 &   898       &   1385         &    673  &     0  &   830   &  1645 \\ \hline
D.yak & 1518   & 1722.5 &   746       &   1253.5       &    461.5&    830 &      0  &  1578 \\ \hline
D.sec & 2041   & 2177   &   1538.5    &   1655.5       &   1463.5&   1645 &   1578  &  0    \\ \hline
\end{tabular}
\begin{center}{Distance matrix for whole genome}\end{center}
\vspace{0.1in}

  \centering
  \includegraphics[bb=0 0 488 412,width=2.56in,height=2.16in,keepaspectratio]{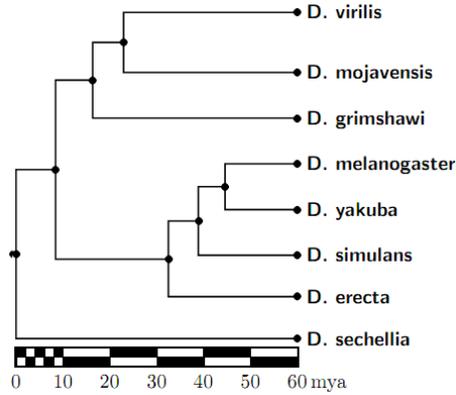}
  \caption{Data, distance matrix and resulting phylogeny for the whole genome} 
  \label{fig:MullerWG} 
\end{figure}

\newpage
\section*{Appendix II: Verification of claimed properties of pointer lists.}

In this appendix we give, for the sake of completeness, the mathematical proofs of the properties of pointer lists stated in Section 3 and used in Section 4\footnote{Some of the results of this section and Section 4 may be deducible from the work presented in Chapters 6 through 9 of \cite{EHPPR}, but we opted for a self-contained presentation of our work.}.

Stipulation (6) of Definition \ref{def: PointerList} will also be called the \emph{exclusion property} below. In our discussion below, the item $x_i$ with minimal absolute value will be called ``the minimal element" of the pointer list, while the item $x_j$ with largest absolute value will be called ``the maximal element" of the pointer list. We shall use the symbol $E$ to denote the entry with smallest absolute value, and $E^{\prime}$ to denote the entry with largest absolute value.

We will also adopt the following terminology for expository ease:
\begin{definition} Two items $x_i$ and $x_j$ in a pointer list are
\begin{enumerate}
  \item{a \emph{pair} if $\vert i-j\vert =1$, and $\min\{i,\, j\}$ is odd.} 
  \item{\emph{mates} if $\vert x_i\vert = \vert x_j\vert$.} 
\end{enumerate}
\end{definition}

\begin{lemma}\label{parityandvalue} If $x_i$ and $x_j$ are items in a pointer list and $i$ and $j$ are distinct but have the same parity, then $x_i \neq x_j$.
\end{lemma}
\begin{proof}
Let $i$ and $j$ be distinct elements of $\{1,\,2,\,\cdots,, n\}$, but have the same parity. Towards deriving a contradiction, assume that contrary to the claim $x_i = x_j$.

{\flushleft\underline{Both i and j are odd.}} 
Then by stipulations (4) and (5) we find that $x_i < x_{i+1}$ and $x_j < x_{j+1}$, and all these items are of the same sign. But then stipulation (6) implies that $x_{i+1} = x_{j+1}$. Let $A$ be $x_i (= x_j)$ and let $B$ be $x_{i+1}(=x_{j+1})$. Then our pointer list is of the form 
\[
 \lbrack\cdots,\, A,\, B,\,\cdots,\,A,\, B,\,\cdots\rbrack
\]
and no pointer has absolute value strictly between the absolute values of $A$ or $B$. Since there is a unique pointer of minimal absolute value and there are two pointers of value $A$, $A$ is not the pointer of minimal absolute value. It follows that an odd number of pointers have their absolute values below that of $A$ (the minimal element and then pairs of mates). Similarly the unique pointer of maximal absolute value is not $B$, and there is an odd number of pointers with absolute values exceeding that of $B$. Since these pointers also all occur in pairs, for some pair $x_k,\, x_{k+1}$ one pointer has absolute value smaller than that of $A$, while the other has absolute value larger than that of $B$, constituting a violation of stipulation (6).

{\flushleft\underline{Both i and j are even.}} 
Then by stipulations (4) and (5) we find that $x_{i-1} < x_i$ and $x_{j-1} < x_j$, and all these items are of the same sign. But then stipulation (6) implies that $x_{i-1} = x_{j-1}$. Let $A$ be $x_{i-1} (= x_{j-1})$ and let $B$ be $x_{i}(=x_{j})$. As before our pointer list is of the form 
\[
 \lbrack\cdots,\, A,\, B,\,\cdots,\,A,\, B,\,\cdots\rbrack
\]
and no pointer has absolute value strictly between the absolute values of $A$ or $B$. By the same argument as in the case when $i$ and $j$ were odd, we now derive a violation of stipulation (6).
\end{proof}

\begin{lemma}\label{valueandparity} For all distinct $i$ and $j$ for which $\vert x_i\vert = \vert x_j\vert$, the following are equivalent:
\begin{enumerate}
  \item{$x_i$ and $x_j$ have the same sign.}
  \item{$i$ and $j$ have distinct parity.}
\end{enumerate}
\end{lemma}
\begin{proof}
{\flushleft{(1) implies (2):}} Assuming $x_i=x_j$, the contrapositive of the implication in Lemma \ref{parityandvalue} gives that $i$ and $j$ have distinct parity.\\
{\flushleft{(2) implies (1):}} We are now assuming that $\vert x_i\vert = \vert x_j\vert$, and that $i$ and $j$ have distinct parity. Suppose that, contrary to (1), $x_i$ and $x_j$ are of opposite sign. We may assume, without loss of generality, that $i$ is odd (and thus $j$ is even). By stipulations (4) and (5) we have that $x_i < x_{i+1}$ and $x_{j-1}<x_j$.
{\flushleft{Case 1: $x_i$ is positive.}} Then we have $\vert x_i\vert <\vert x_{i+1}\vert$, and $\vert x_j\vert < \vert x_{j-1}\vert$. As $\vert x_i\vert = \vert x_j\vert$, stipulation (6) implies that $\vert x_{i+1}\vert = \vert x_{j-1}\vert$. Letting $A$ denote $\vert x_i\vert (=\vert x_j\vert)$, and letting $B$ denote $\vert x_{i+1}\vert (=\vert x_{j-1}\vert)$, we find that as in the proof of Lemma \ref{parityandvalue} we have a pointer list of the form
\[
 \lbrack\cdots,\, A,\, B,\,\cdots,\,-B,\, -A,\,\cdots\rbrack
\]
where no pointer has absolute value strictly between $A$ and $B$. As before an odd number of pointers have absolute value less than $A$, and an odd number have absolute value larger than $B$, and some pointer pair includes one from each of these two categories, constituting a violation of stipulation (6).
{\flushleft{Case 2: $x_i$ is negative.}} Similar considerations show a violation of stipulation (6).
\end{proof}

\begin{lemma}\label{edgeslemma}
There are no pointer lists of the form $\lbrack B,\, C,\, \cdots,\, C,\,B\rbrack$.
\end{lemma}
\begin{proof}
Note that by stipulation (1) of the pointer list definition the parity of the leftmost instance of $B$ is \emph{odd}, while the parity of the rightmost instance of $B$ is \emph{even}. The reverse applies to the two instances of $C$. By stipulation (5) of the pointer list definition we would have $B\le C$ and $C\le B$, whence there are four copies of $B$ in the pointer list, violating stipulation (4) for pointer lists.  
\end{proof}

\begin{lemma}\label{CentralLemma}
If $\lbrack x_1,\, x_2,\,\cdots,x_{m-1},\,x_m\rbrack$ is a pointer list of length larger than $4$, then at least one of the following three statements is false:
\begin{enumerate}
\item[(a)] $(\forall i)( x_i \neq x_{i+1})$
\item[(b)] $(\forall i)(\forall j)(\mbox{If } \vert x_i\vert = \vert x_j\vert, \mbox{ then } x_i = x_j)$
\item[(c)] $(\forall i)(\forall j)(\forall k)(\forall \ell)(\mbox{If }i \neq k,\, j\neq \ell,\, i < j \mbox{ and }x_i = x_k\mbox{ and }x_j = x_{\ell}, \mbox{ then either }i < j < \ell < k \mbox{ or }i < k < j < \ell)$
\end{enumerate}
\end{lemma}
\begin{proof}
Assume that contrary to the claim of the lemma, there exists a pointer list of length larger than 4 which also satisfies properties (a), (b) and (c). 

By (b) any two pointer entries that are mates are of the same sign.  Thus by Lemma \ref{valueandparity} the positions in which mates occur are of opposite parity. By (a), no mates form a pair. By (c), the only relative configurations possible between two sets of mates are 
\begin{equation}\label{disjoint}
 \lbrack \cdots A \cdots A \cdots B \cdots B \cdots \rbrack
\end{equation}
and 
\begin{equation}\label{nested}
 \lbrack \cdots A \cdots B \cdots B \cdots A \cdots \rbrack
\end{equation}

{\flushleft{\bf Sublemma A:}} Configurations of the form $\lbrack \cdots A, \cdots A, \cdots B,\cdots B, \cdots C, \cdots C,\cdots \rbrack
$ are impossible: 
{\flushleft{\bf Proof of Sublemma A: }} 
For by (a) there are $x_A$, $x_B$ and $x_C$ with $x_A$ between the two copies of $A$, $x_B$ between the two copies of $B$ and $x_C$ between the two copies of $C$. We may assume that we have selected the $A$, $B$ and $C$ for which the number of items between consecutive copies of the same symbol is minimal in each case. Thus, there is no $D$ such that both copies of $D$ are between the two copies of $A$, or of $B$ or of $C$. But at least one of $x_A$, $x_B$ or $x_C$ differs from the minimal and from the maximal element of the pointer list, and thus has a partner symbol (a mate) (of the same sign, by (b)) located in the pointer list. Assume it is $x_A$ (the argument is the same for the other cases): Then the other copy of $x_A$ does not occur between the two copies of $A$. But then the two copies of $A$ and the two copies of $x_A$ constitute a violation of (c). This completes the proof of Sublemma A.$\Box$

It follows that if we have a pointer list satisfying (a), (b) and (c), then for all distinct triples $A$, $B$ and $C$ that are not the minimal or maximal elements of the pointer list we have configurations of only the following two general forms:
\begin{equation}\label{onenest}
  \lbrack \cdots A, \cdots B, \cdots C, \cdots C, \cdots B, \cdots A, \cdots \rbrack \mbox{ or}
\end{equation}
\begin{equation}\label{twonests}
  \lbrack \cdots A, \cdots C,\cdots C,\cdots A, \cdots B,\cdots,\, B,\, \cdots \rbrack,
\end{equation}

{\flushleft{\bf Sublemma B:}} No pointer list of length larger than 4 is of any of the forms\\
(i) $\lbrack x,\,\cdots,\, y\rbrack$, \\
(ii) $\lbrack x,\, y,\,\cdots\, \rbrack$,\\
(iii) $\lbrack \cdots,\, x,\, y \rbrack$,\\
where $\{\vert x\vert,\, \vert y\vert\} = \{E,\, E^{\prime}\}$.\\
{\flushleft{\bf Proof of Sublemma B:}}
For suppose some pointer list is of one of these forms. Choose a pair of mates, say of value $A$, with the fewest possible pointers between them. Thus, we have a pointer list of one of these three forms, which contains a pattern $\lbrack\cdots,\, A,\,\cdots,\, A,\,\cdots\rbrack$ and the number of pointers between the two copies of $A$ is as small as possible. Suppose $B$ is a pointer appearing between these two $A$'s. Since $B$ is not $E$ or $E^{\prime}$ we have by property (c) two copies of $B$ appearing between these two $A$'s, contradicting the minimality condition on the number of pointers between two mates. This concludes the proof of Sublemma B. 

{\flushleft{\bf Sublemma C:}} If there is a pointer list of form $\lbrack A,\, \cdots,\, B\rbrack$ where $A$ and $B$ each has mates, then it is of the form $\lbrack A,\,\cdots,\, A,\,B,\,\cdots,\,B\rbrack$.\\
{\flushleft{\bf Proof of Sublemma C:}}
The mate of the initial $A$ must be in a position $i$ which is even, and the mate of the terminal $B$ must be in a position $j$ which is odd. By property (c) we have $i<j$. Because $i$ is even and $j$ is odd, if $j$ is not $i+1$, then there are a positive even number of pointers between the second copy of $A$, and the first copy of $B$. These pointers cannot have absolute value $E$ or $E^{\prime}$, for otherwise there would be between the two copies of $A$, or else between the two copies of $B$, a configuration of the form $\cdots,\,C,\,C,\cdots$, which is forbidden by (a). But then there is some other pointer $C$ between the $A$ in position $i$ and the $B$ in position $j$. By (c) we must have the mate of $C$ also between the $A$ in position $i$ and the $B$ in position $j$, meaning the pointer list is of the form
$\lbrack A,\,\cdots,\,A,\,\cdots,\, C,\,\cdots,\, C,\,\cdots,\, B,\,\cdots,\, B\rbrack$,
and this contradicts Sublemma A. This concludes the proof of Sublemma C.

{\flushleft{\bf Sublemma D}} Configurations as in (\ref{onenest}) are impossible. \\

{\flushleft{\bf Proof of Sublemma D: }} By Lemma \ref{edgeslemma} only the following are possibilities for (\ref{onenest}):
\[
\mbox{(a) }  \lbrack E,\, B,\, C,\,\cdots,\, A,\, E^{\prime},\, A,\, \cdots,\, C,\, B\rbrack
\]
\[
\mbox{(b) }  \lbrack B,\, C, \,\cdots,\, A,\, E^{\prime},\, A,\, \cdots,\, C,\, B,\, E\rbrack
\]
\[
\mbox{(c) }  \lbrack B,\, E,\, C,\,\cdots,\, A,\, E^{\prime},\, A,\, \cdots,\, C,\, B\rbrack
\]
\[
\mbox{(d) }  \lbrack B,\, C,\,\cdots,\, A,\, E^{\prime},\, A,\, \cdots,\,D,\, C,\, E,\, B\rbrack
\]
But then the parity of the positions of the two copies of $B$ in (a) and in (b) are the same, while the parity of the positions of $C$ in (c) and (d) are the same. This contradicts Lemma \ref{valueandparity}. This completes the proof of Sublemma D. $\Box$

{\flushleft{\bf Sublemma E}} Configurations as in (\ref{twonests}) are impossible.\\
{\flushleft{\bf Proof of Sublemma E}}
Consider configuration (\ref{twonests}): 
\[
  [\dots,\, A,\, \dots,\, C,\,\dots,\,C,\, \dots A,\, \dots,\, B,\, \dots,\, B,\, \dots].
\]
Neither $\vert x_1\vert$, nor $\vert x_n\vert$ can be a member of $\{E,\, E^{\prime}\}$, since this will allow a configuration of the form $\dots\,D,\,D,\,\dots$ occurring between the two copies of $A$ or the two copies of $B$, contradicting (a). By Sublemma C configurations as in (\ref{twonests}) must be of the form
$\lbrack A,\,\cdots,\,C,\,\cdots,\,C,\cdots,\,A,\,B,\,\cdots,\,B\rbrack$. To avoid a contradiction with premise (a), this configuration must be of the form
\[
  \lbrack A,\,\cdots,\,C,\,\cdots,x_i,\,\cdots\,C,\cdots,\,A,\,B,\,\cdots,\,x_j,\,\cdots,\,B\rbrack.
\]
where $\{\vert x_i\vert,\,\vert x_j\vert\} = \{E,\, E^{\prime}\}$. Applying premise (a) again we see that for each pointer $D$ between positions 1 and $i$, its mate is in the corresponding position between positions $i$ and the position of the mate of $A$. The same remark applies to the segment between the two $B$'s of the pointerlist. Thus the pointer list is of the form
\[
 \lbrack A_1,\, A_2,\,\cdots,\,A_k,\,x_i,\, A_k,\,\cdots,\, A_2,\, A_1,\, B_1,\, B_2,\,\cdots,\,B_t,\,x_j,\, B_t,\, \cdots,\, B_2,\, B_1\rbrack
\] 
But then both copies of $A_1$ are in odd positions, contradicting Lemma \ref{parityandvalue}.
This completes the proof of Sublemma E, and thus of Lemma \ref{CentralLemma}.
\end{proof}

\begin{center}{\bf Examples of pointer lists.}\end{center}

Let $\integers$ denote the set of integers. for an integer $z$ we define
\[
   \check{z}(1) =\begin{cases} z   & \mbox{if } z = \vert z\vert \\
                               z-1 & \mbox{otherwise } 
         \end{cases}
\] 
and  in all cases $\check{z}(2) = z(1)+1$. 

For a set $S$ the symbol $^{<\omega}S$ denotes the set of finite sequences with entries from $S$. 
Define the function $\pi:\,^{<\omega}\integers\rightarrow \,^{<\omega}\integers$ by:
\[
  \pi(\lbrack z_1,\cdots,z_k\rbrack) = \lbrack \check{z}_1(1),\check{z}_1(2),\cdots,\check{z}_k(1),\check{z}_k(2)\rbrack
\]
Thus, for example, $\pi(\lbrack -1,4,3,5,2,-9,7,10,-8,6\rbrack)$ is the sequence 
\[
  \lbrack -2,-1,4,5,3,4,5,6,2,3,-10,-9,7,8,10,11,-9,-8,6,7\rbrack.
\]

\begin{lemma}\label{pointerassignment}
For each finite sequence $M:=\lbrack s_1,\, s_2,\, \cdots, \, s_n\rbrack$ of non-zero integers such that there is an integer $m$ for which $\{\vert s_i\vert:1\le i\le n\} = \{m+1,\, \cdots,\,m+n\}$, the sequence $\pi(M)$ is a pointer list.   
\end{lemma}
\begin{proof}
The sequence $\lbrack \vert s_1\vert,\, \vert s_2\vert,\, \cdots,\,\vert s_n\vert\rbrack$ is a permutation of the numbers $m+1$ through $m+n$. Note that $m\ge 0$.

From the definition of $\pi$ we have for each $j$ that 
\[
  \check{s}_j(1), \check{s}_j(2) = \left\{\begin{tabular}{ll}
                                          m+i, m+i+1       & if $s_j$ = m+i    \\
                                          -(m+i)-1,-(m+i)  & if $s_j$ = - (m+i)\\ 
                                      \end{tabular}
                               \right. 
\]
Thus for each odd indexed entry $x_j$ in $\pi(M)$, the absolute values $\vert x_j\vert$ and $\vert x_{j+1}\vert$ are successive positive integers, meaning that stipulation (6) in the definition of a pointer list is satisfied by $\pi(M)$. It is also evident that stipulation (1) is satisfied. Also note that the smallest absolute value obtained by terms of $\pi(M)$ is $m+1$, and this is achieved by exactly one entry of $\pi(M)$. Similarly, the largest absolute value achieved is $m+n+1$, and is achieved by exactly one term in $\pi(M)$. Thus stipulations (2) and (3) are satisfied. Towards stipulation (4), consider an entry $x_i$ of $\pi(M)$ which is not of least or largest absolute value. Note that then we have $m+2\le m+t =  \vert x_i\vert\le m+n$. Choose $j$ such that $x_i = \check{s}_j(1)$, or $x_i=\check{s}_j(2)$. Find the $k$ for which $\vert s_k\vert = m+t-1$, and also find the $\ell$ for which $\vert s_{\ell}\vert = m+t+1$. Then $\vert x_i\vert$ is equal to exactly one of $\check{s}_k(1),\, \check{s}_k(2),\, \check{s}_{\ell}(1)$ or $\check{s}_{\ell}(2)$. Thus, stipulation (4) is satisfied. To see stipulation (5), observe that for any odd $i$, $\{x_i,\, x_{i+1}\} = \{\check{s}_j(1), \check{s}_{j}(2)\}$, and thus these two entries have the same sign. \end{proof}

\newpage
\section*{Appendix III: Verification of claimed properties of ciliate operations on pointer lists.}

\begin{theorem}\label{domain} If $P$ is a pointer list of length larger than $4$, then at least one of the following statements is true:
\begin{enumerate}
  \item ${\sf cde}(P)\neq P$;
  \item ${\sf cdr}(P)\neq P$;
  \item ${\sf cds}(P)\neq P$.
\end{enumerate}
\end{theorem}
\begin{proof}
Let $P = \lbrack x_1,\, x_2,\,\cdots,\,x_{m-1},\, x_m\rbrack$ be a pointer list. By Lemma \ref{CentralLemma} at least one of the following three statements is false:
\begin{enumerate}
\item[(a)] $(\forall i)( x_i \neq x_{i+1})$
\item[(b)] $(\forall i)(\forall j)(\mbox{If } \vert x_i\vert = \vert x_j\vert, \mbox{ then } x_i = x_j)$
\item[(c)] $(\forall i)(\forall j)(\forall k)(\forall \ell)(\mbox{If }i \neq k,\, j\neq \ell,\, i < j \mbox{ and }x_i = x_k\mbox{ and }x_j = x_{\ell}, \mbox{ then either }\\ i < j < \ell < k \mbox{ or }i < k < j < \ell)$
\end{enumerate}

{\flushleft{\tt If statement (a) fails:}} Then for some $i$ we have $x_i = x_{i+1}$. For the minimal such $i$, ${\sf cde}(P) = \lbrack x_1,\cdots,x_{i-1},\,x_{i+2},\,\cdots,\, x_m\rbrack \neq P$.

{\flushleft{\tt If statement (b) fails:}} Fix the minimal $i $ for which there is a $j>i$ with $\vert x_i\vert = \vert x_j\vert$, but $x_i\cdot x_j<0$. Then 
\[
  {\sf cdr}(P) = \lbrack x_1,\cdots, x_{i},-x_j,\, -x_{j-1},\, \cdots,\, -x_{i+1},\, x_{j+1},\,\cdots,\, x_m\rbrack\neq P.
\] 

{\flushleft{\tt If neither (a) nor (b) fails:}} Then (c) fails. Fix $i$, $j$, $k$ and $\ell$ witnessing this failure. We may assume that $i<j$ and $x_i=x_k$, and $x_j=x_{\ell}$, and that $i\neq k$ and $j\neq \ell$. By stipulation (4) in the definition of a pointer list we have $x_i\neq x_j$, and thus $k\neq \ell$. Since (c) fails, the two configurations claimed by (c) are false for our witness. We have that $i<j<k<\ell$, or $i<\ell < k < j$, or $\ell<i<j<k$. In each case an application of cds results in a sequence ${\sf cds}(P)\neq P$. 
\end{proof}

\begin{theorem}[Pointer list preservation]\label{parityth}
Let $P = \lbrack x_1,\cdots,x_m\rbrack$ be a pointer list.
Then each of ${\sf cde}(P)$, ${\sf cdr}(P)$ and ${\sf cds}(P)$ is a pointer list.
\end{theorem}
\begin{proof}
We first verify that the equivalence of Lemma \ref{valueandparity} is preserved by an application of {\sf cde}, {\sf cdr} or {\sf cds} to a pointer list.

{\flushleft{\underline{Operation {\sf cde}: }}} We need to consider only the case when ${\sf cde}(P)\neq P$. Fix the smallest $i$ such that $x_i = x_{i+1}$. Then we have 
\[
   {\sf cde}(P) = \lbrack x_1,\, \cdots,\,x_{i-1},x_{i+2 },\cdots,x_m\rbrack.
\]
The parity of the position of each surviving term is the same as before, since the position number changed by 0 or by 2. Since {\sf cde} does not affect the signs of the terms in the original pointer list, the equivalence of Lemma \ref{valueandparity} still holds for ${\sf cde}(P)$.

{\flushleft{\underline{Operation {\sf cdr}: }}} Now assume that ${\sf cde}(P)=P$ and ${\sf cdr}(P)\neq P$. Choose the least $i$ such that for a $j>i$ we have $x_i = -x_j$. By Lemma \ref{valueandparity}, $i+j$ is even. Now 
\[
{\sf cdr}(P) = \lbrack x_1,\cdots,x_i, -x_j,\, -x_{j-1},\, \cdots,\, -x_{i+1},\, x_{j+1},\,\cdots,\, x_m\rbrack,
\]
and either $i$ and $j$ are even, or else $i$ and $j$ are odd. Thus there are an even number of terms moved and signs changed in this application of ${\sf cdr}$. Each of these terms also is moved to a position whose position number is of opposite parity of the original position number. Thus, the equivalence of Lemma \ref{valueandparity} still holds of ${\sf cdr}(P)$.

{\flushleft{\underline{Operation \sf cds: }}}  We may assume that ${\sf cde}(P) = P$. Suppose that ${\sf cds}(P) \neq P$, and choose the lexicographically least $(i,j,k,\ell)$ such that $i<j<k<\ell$ and $x_i = x_k$ and $x_j = x_{\ell}$.
By Lemma \ref{valueandparity} $i$ and $k$ have opposite parity, and $j$ and $\ell$ have opposite parity. Then {\sf cds}(P) is equal to
\[
\lbrack x_1,\cdots,x_{i-1},\,{\color{red}x_i},\,{\color{blue}x_k,\,\cdots,\,x_{\ell}},\,{\color{red}x_j},\, x_{j+1},\,\cdots,\,x_{k-1},\,{\color{red}x_{i+1},\,\cdots,\,x_{j-1}},\, x_{\ell+1},\,\cdots,x_m\rbrack.
\]
{\flushleft{\bf Case 1: } $i$ is even and $j$ is odd.} By Lemma \ref{valueandparity} we then also have that $k$ is odd and $\ell$ is even. Thus, an even number of blue terms are swapped with an even number of red terms, and the parities of the positions of all terms remain the same. Since no signs are changed during an application of {\sf cds}, the equivalence of Lemma \ref{valueandparity} remains true of ${\sf cds}(P)$.
{\flushleft{\bf Case 2:} $i$ is even and $j$ is even.} By Lemma \ref{valueandparity} we have that $k$ and $\ell$ are both odd. In this case an odd number of blue terms are swapped with an odd number of red terms, and no signs are changed. Since the terms's positions have the same parities as before, it follows the equivalence of Lemma \ref{valueandparity} still holds of ${\sf cds}(P)$.

The cases when $i$ is odd and $j$ even, or when $i$ is odd and $j$ is odd, use similar arguments.

What remains to be proved is that the result of applying any of {\sf cde}, {\sf cdr} or {\sf cds} to the pointer list $P$ is again a pointer list. 

Since neither of ${\sf cdr}$ or ${\sf cds}$ changes the number of terms of the list, and since ${\sf cde}$ deletes exactly two terms or none, the result has an even number of terms. Thus stipulation (1) in the definition of a pointer list is preserved. Since the terms least and largest in absolute value are unique, and only {\sf cde} removes consecutive terms that are equal, these two terms survive all applications of {\sf cde}, {\sf cdr} or {\sf cds}. Thus stipulations (2) and (3) in the definition of pointer lists is preserved by these operations.
 
Since only {\sf cde} removes terms that are adjacent and equal, and since none of the operations {\sf cde}, {\sf cdr} or {\sf cds} affects the absolute value of any term, also stipulation (4) in the definition of pointer lists is preserved by these operations.

We must verify stipulations (5) and (6).\\ 
\underline{Consider ${\sf cde}(P)$:} Suppose $i$ is minimal with $x_i = x_{i+1}$.\\
Case 1: i is odd:\\
Then ${\sf cde}(P) = \lbrack x_1,\,\,x_2,\,\,x_{i-1},\,\,x_{i+2},\,\,\cdots x_{m-1},\,\, x_m\rbrack$. Stipulation (6) remains true since the removal of the two consecutive terms do not change the parity of the remaining indices, and thus does not affect the truth of stipulation (6) for the remaining terms. The same reason shows that stipulation (5) remains true for ${\sf cde}(P)$.\\
Case 2: i is even:\\
Now by stipulation (5) we see that $x_{i-1}\le x_i = x_{i+1}\le x_{i+2}$, and these terms have the same sign. Upon applying ${\sf cde}$, we have $\lbrack x_1,\,\, \cdots,\,\, x_{i-1},\,\, x_{i+2},\,\,\cdots,\,\, x_m\rbrack$ and all stipulations of the definition of pointer list are still satisfied. We verify stipulation (6): By the exclusion property there are no $x_j$ with absolute value between the absolute values of $x_{i-1}$ and $x_i$, and no $x_j$ with absolute value between the absolute values of $x_{i+1}$ and $x_{i+2}$. Thus upon the removal of $x_i = x_{i+1}$, there is no $x_j$ with absolute value between the absolute values of $x_{i-1}$ and $x_{i+2}$. It follows that ${\sf cde}(P)$ still satisfies stipulation (6).

{\flushleft{\underline{Consider ${\sf cdr}(P)$:}}} Suppose $i$ is minimal such that for a $j>i$ we have $x_i = -x_{j}$. Application of ${\sf cdr}$ to $P$ yields
\[
{\sf cdr}(P) = \lbrack x_1,\cdots,x_i, {\color{blue}-x_j,\, -x_{j-1},\, \cdots,\, -x_{i+1}},\, x_{j+1},\,\cdots,\, x_m\rbrack.
\]
If $i$ is even then by Lemma \ref{valueandparity} $j$ is even and so $x_{j-1} < x_j$, implying that $-x_j< -x_{j-1}$, and by stipulation (6) there is no term from ${\sf cdr}(P)$ with absolute value between the absolute values of $-x_j$ and $-x_{j-1}$. Since $i$ is even we similarly have $x_{i-1}<x_i$ and there are no terms in ${\sf cdr}(P)$ with absolute value between the absolute values of $x_{i-1}$ and $x_i$. Also, as $i$ is even $i+1$ is odd, and so $-x_{i+1}$ is in an even parity position, and there still are no terms of ${\sf cdr}(P)$ with absolute value between the absolute values of $-x_{i+2}$ and $-x_{i+1}$.\\
If $i$ is odd, then by Lemma \ref{valueandparity} $j$ is odd. Thus $-x_{i+1}\le -x_i = x_j\le x_{j+1}$ and by stipulation (6) there are no terms of $P$ in absolute value between $\vert x_i\vert$ and $\vert x_{i+1}\vert$. Similarly there are none with absolute value in the interval $\vert x_j\vert$ and $\vert x_{j+1}\vert$. But then, aside of $x_i=-x_j$ there are no terms of ${\sf cdr}(P)$  with absolute values between $\vert x_{i+1}\vert$ and $\vert x_{j+1}\vert$. Since atipulation (6) for the other indices is not affected by ${\sf cdr}$ if follows that ${\sf cdr}(P)$ has still satisfies stipulation (6). Parity and sign arguments show that ${\sf cdr}(P)$ still meets stipulation (5) of the pointer list definition.

{\flushleft{\underline{Consider ${\sf cds}(P)$:}}} Choose the lexicographically least $(i,j,k,\ell)$ such that $i<j<k<\ell$ and $x_i = x_k$ and $x_j = x_{\ell}$. Then {\sf cds}(P) is
\[
\lbrack x_1,\cdots,x_{i-1},\,{\color{red}x_i},\,{\color{blue}x_k,\,\cdots,\,x_{\ell}},\,{\color{red}x_j},\, x_{j+1},\,\cdots,\,x_{k-1},\,{\color{red}x_{i+1},\,\cdots,\,x_{j-1}},\, x_{\ell+1},\,\cdots,x_m\rbrack.
\]
To verify stipulations (5) and (6) for ${\sf cds}(P)$, given $P$ satisfies stipulations (5) and (6), we argue as follows:

{\flushleft{\bf Case 1: } $i$ is even and $j$ is odd.} By Lemma \ref{valueandparity} $k$ is odd and $\ell$ is even. Thus, an even number of blue terms are swapped with an even number of red terms, and the parities of the positions of all terms remain the same and no signs are changed during an application of {\sf cds}. Since the parities are preserved, stipulation (5) is preserved. To see that stipulation (6) is preserved observe that no pairs of the form $x_t,\, x_{t+1}$ with $t$ of odd parity are disrupted by this instance of {\sf cds}. 

{\flushleft{\bf Case 2:} $i$ is even and $j$ is even.} By Lemma \ref{valueandparity} $k$ and $\ell$ are both odd. In this case an odd number of blue terms are swapped with an odd number of red terms, and no signs are changed. The only pair of the form $x_t,\, x_{t+1}$ with $t$ odd that is disrupted is the case when $t=\ell$. Since $\ell$ is odd and $j$ is even, $j-1$ is odd and we have $x_{j-1}\le x_j = x_{\ell}\le x_{\ell+1}$. It follows that stipulation (5) is still true, and that stipulation (6) still holds of ${\sf cds}(P)$.
\end{proof}

\end{document}